\documentclass[smallextended]{svjour3}
\usepackage[T1]{fontenc}
\smartqed

\usepackage{graphicx}
\usepackage{amsfonts}
\usepackage{amsmath}
\usepackage[noadjust]{cite}
\usepackage{amssymb}
\usepackage{enumitem}
\usepackage[]{todonotes}
\usepackage{comment}
\usepackage{hyperref}
\usepackage{algorithm}
\usepackage{algorithmic}

\newcommand{\Bag}{B}

\def\eps{\varepsilon}
\def\tq{\tilde{q}}

\newtheorem{observation}[theorem]{Observation}
\newtheorem{lemm}[theorem]{Lemma}
\newtheorem{coro}[theorem]{Corollary}
\usepackage{tabularx, environ}

\makeatletter

\newcolumntype{\expand}{}
\long\@namedef{NC@rewrite@\string\expand}{\expandafter\NC@find}

\NewEnviron{fancyproblem}[2][]{%
	\def\problem@arg{#1}%
	\def\problem@framed{framed}%
	\def\problem@lined{lined}%
	\def\problem@doublelined{doublelined}%
	\ifx\problem@arg\@empty%
	\def\problem@hline{}%
	\else%
	\ifx\problem@arg\problem@doublelined%
	\def\problem@hline{\hline\hline}%
	\else%
	\def\problem@hline{\hline}%
	\fi%
	\fi%
	\ifx\problem@arg\problem@framed%
	\def\problem@tablelayout{|>{\bfseries}lX|c}%
	\def\problem@title{\multicolumn{2}{|l|}{%
			\raisebox{-\fboxsep}{\textsc{\large #2}}%
	}}%
	\else
	\def\problem@tablelayout{>{\bfseries}lXc}%
	\def\problem@title{\multicolumn{2}{l}{%
			\raisebox{-\fboxsep}{\textsc{\large #2}}%
	}}%
	\fi%
	\bigskip\par\noindent%
	\begin{tabularx}{\textwidth}{\expand\problem@tablelayout}%
		\problem@hline%
		\problem@title\\[2\fboxsep]%
		\BODY\\\problem@hline%
	\end{tabularx}%
	\medskip\par%
}
\makeatother

\begin{document}
\vspace{-1cm}
\title{Fair allocation of indivisible items with conflict graphs\thanks{A preliminary version containing some of the results presented here appeared in~\cite{iwoca2020}.}}

	\author{Nina Chiarelli
		\and
		Matja\v z Krnc
		\and
		Martin Milani\v c
		\and
		Ulrich Pferschy
		\and
		Nevena Piva\v c
		\and
		Joachim Schauer}
	\authorrunning{N.~Chiarelli et al.}
	
	\institute{N. Chiarelli \at University of Primorska, FAMNIT and IAM, Koper,
	Slovenia,
		\email{nina.chiarelli@famnit.upr.si}
	    \and M. Krnc \at University of Primorska, FAMNIT and IAM, Koper, Slovenia,
	    \email{matjaz.krnc@upr.si}
		\and M. Milani\v c \at University of Primorska, FAMNIT and IAM, Koper, Slovenia
		\email{martin.milanic@upr.si}
		\and U. Pferschy \at University of Graz, Austria,
		\email{ulrich.pferschy@uni-graz.at}
		\and N. Piva\v c \at University of Primorska, FAMNIT and IAM, Koper, Slovenia
		\email{nevena.pivac@iam.upr.si}
		\and J. Schauer \at FH JOANNEUM, Austria,
		\email{joachim.schauer@fh-joanneum.at}}
	
	\date{}
	\maketitle
	
	\vspace{-1cm}
    \begin{abstract}
		We consider the fair allocation of indivisible items to several agents and add a graph theoretical perspective to this classical problem.
	Namely,	we introduce an incompatibility relation between pairs of items described in terms of a conflict graph.
	Every subset of items assigned to one agent has to form an independent set in this graph.
	Thus, the allocation of items to the agents corresponds to a partial coloring of the conflict graph.
	Every agent has its own profit valuation for every item.
	Aiming at a fair allocation, our goal is the maximization of the lowest total profit of items allocated to any one of the agents.
	The resulting optimization problem contains, as special cases, both \textsc{Partition} and \textsc{Independent Set}.
	In our contribution we derive complexity and algorithmic results depending on the properties of the given graph.
	We show that the problem is strongly NP-hard for bipartite graphs and their line graphs, and solvable in pseudo-polynomial time for the classes of chordal graphs, cocomparability graphs, biconvex bipartite graphs, and graphs of bounded treewidth.
	Each of the pseudo-polynomial algorithms can also be turned into a fully polynomial approximation scheme (FPTAS).

	\end{abstract}

\section{Introduction}
\label{sec:intro}

Allocating resources to several agents in a satisfactory way is a classical problem in combinatorial optimization.
In particular, interesting questions arise if agents have different valuations of resources or if additional constraints are imposed for a feasible allocation.
In this work we study the fair allocation of $n$ indivisible goods or items to a set of $k$ agents.
Each agent has its own additive utility function over the set of items.
The goal is to assign every item to exactly one of the agents so that the minimal utility over all agents is as large as possible.
Related problems of fair allocation are frequently studied in Computational Social Choice, see, e.g.,~\cite{hand16}.
Recent papers from this field containing many pointers to the literature and studying fairness issues, also in connection with an underlying graph structure, are given by~\cite{BLM21,ijcai2017-20}.
In the area of Combinatorial Optimization a similar problem is well-known as the \emph{Santa Claus} problem (see~\cite{MR2277128}), which can also be seen as a scheduling problem.

In this paper we look at the problem from a graph theoretical perspective and add a major new aspect to it.
We allow an incompatibility relation between pairs of items, meaning that incompatible items should not be allocated to the same agent.
This can reflect the fact that items rule out their joint usage or simply the fact that certain items are identical (or of a similar type) and it does not make sense for one agent to receive more than one of these items.
We will represent such a relation by a \emph{conflict graph} where vertices correspond to items and edges express incompatibilities.

{
As a more concrete example consider the distribution of transportation orders among a number of shipping partners which should all be treated as equally as possible according to a joint master agreement.
In some industries, goods cannot be combined in an arbitrary way due to safety regulations or rules for hazardous materials
(see \cite{Santos19} for the delivery of goods from incompatible categories to small neighborhood stores).
Then, a conflict graph can be used to express forbidden freight combinations (see, e.g.,~\cite{Factorovich20,Hu15}).

When items represent tasks with a starting and end time, each agent should be allocated a fair subset of non-overlapping tasks.
Again, the mutual exclusion of two tasks/items, will be represented by the edges of a conflict graph (see, e.g.,~\cite{Mallek22,raey09}).
Note that in \cite{Brito21} a general treatment of conflict graphs was performed for the COIN OR Branch-and-Cut (CBC) solver\footnote{https://github.com/coin-or/Cbc}.

In all such scenarios} every feasible allocation to one agent must be an {independent set} in the conflict graph. This means that the overall solution can also be expressed as a \emph{partial $k$-coloring} of the conflict graph $G$, but in addition every vertex/item has a profit value for every color/agent and the sum of profits of vertices/items assigned to one color/agent should be optimized in a maxi-min sense.

We believe that this problem combines aspects of independent sets, graph coloring, and weight partitioning in an interesting way, offering new perspectives to look at these classical combinatorial optimization problems.

{\subsection{Problem definitions}}

\noindent{\bf The classical fair division problem.}
We consider a set $V$ of items with cardinality $|V|=n$ and $k$ profit functions $p_1,\ldots,p_k: V \to \mathbb{Z}_+$.
{An \emph{ordered \hbox{$k$-partition}} of $V$ is  a sequence $(X_1,\ldots,X_k)$ of $k$ pairwise disjoint subsets of $V$ such that $\bigcup_{i = 1}^k X_i = V$.}
The \emph{satisfaction level} of an ordered $k$-partition $(X_1,\ldots,X_k)$ of $V$ (with respect to $p_1,\dots, p_k$) is defined as the minimum of the resulting profits $p_j(X_j) := \sum_{v\in X_j}p_j(v)$, where $j\in \{1,\ldots, k\}$.
The classical fair division problem can be stated as follows.

\medskip

\noindent\parbox{0.82\linewidth}{\noindent
	{\sc Fair $k$-Division of Indivisible Items}\\[1.2ex]
	\begin{tabular*}{.8\textwidth}{ll}
		\noindent{\bf Input:} & A set $V$ of $n$ items, $k$ profit functions $p_1,\ldots,p_k: V \to \mathbb{Z}_+$.\\[0.5ex]
		\noindent{\bf Task:} & Compute an ordered $k$-partition of $V$ with maximum satisfaction\\ &  level.
	\end{tabular*}
}

\medskip
\noindent{\bf Connections with scheduling and knapsack problems.}
For the special case where all $k$ profit functions are identical, i.e., $p_1= p_2=\ldots =p_k$,
the problem can also be represented in a scheduling setting. There are $k$ identical machines and $n$ jobs, which have to be assigned to the machines by a $k$-partitioning. The goal is to maximize the minimal completion time (corresponding to the satisfaction level) over all $k$ machines.
It was pointed out in \cite{DFL1982} that this problem is weakly NP-hard even for $k=2$ machines.
Indeed, it is easy to see that an algorithm deciding the above scheduling problem for two machines would also decide the classical \textsc{Partition} problem: given $n$ integers $a_1,\ldots, a_n$, can they be partitioned into two subsets with equal sums? For $k\geq 3$, one can simply add jobs of length one half of the sum of weights in the instance of \textsc{Partition}.
If $k$ is not fixed, but part of the input, the same scheduling problem is strongly NP-hard {as mentioned in \cite{azep98} (a~PTAS was derived in~\cite{MR1452078})}.
In fact, an instance of the strongly NP-complete \textsc{3-Partition} problem with $3m$ elements and target bound $B$ could be decided by any algorithm for the scheduling problem
with $n=3m$ jobs, $k=m$ machines and a desired minimal completion time equal to $B$.
We conclude for later reference.

\medskip
\begin{observation}\label{obs:SantaNP}
	\textsc{Fair $k$-Division of Indivisible Items}, even with $k$ identical profit functions, is weakly NP-hard for any constant $k\geq 2$ and strongly NP-hard for $k$ being part of the input.
\end{observation}

{
Note that for $k = 2$, the decision version of	{\sc Fair $k$-Division of Indivisible Items} also generalizes the decision version of the {\sc Knapsack} problem:
Given a set $V = \{1,\ldots, n\}$ of items with weights $w_1,\ldots, w_n\in \mathbb{Z}_+$ and values $v_1,\ldots, v_n\in \mathbb{Z}_+$, and two positive integers $W$ and $C$ such that $W<\sum_{j\in V}w_j$, is there a subset of the items having total weight at most $W$ and total value at least $C$?
\footnote{Indeed, by considering two profit functions $p_1,p_2:V\to \mathbb{Z}_+$ defined by $p_1(i) = \Delta\cdot v_i$ where $\Delta = \sum_{j\in V}w_j-W$ and $p_2(i) = C\cdot w_i$ for all $i\in V$, it is not difficult to verify that such a set $S$ exists if and only if $V$ admits an ordered $2$-partition with satisfaction level at least $C\cdot \Delta$.} }

{
It should be noted that \textsc{Fair $k$-Division of Indivisible Items}} is still only weakly NP-hard for constant $k$ even for arbitrary profit functions, since we can construct a pseudo-polynomial algorithm solving the problem with a $k$-dimensional dynamic programming array.

\bigskip
\noindent{\bf Our generalization.}
In this paper we study a generalization of \textsc{Fair $k$-Division of Indivisible Items}, where a \emph{conflict graph} $G = (V,E)$ on the set $V$ of items to be divided is introduced.
An edge $\{i,j\} \in E$ means that items $i$ and $j$ should not be assigned to the same subset of the partition.
Allocating items in a conflict-free way immediately gives rise to (partial) colorings of the graph, a concept studied by Berge~\cite{MR989117} and de Werra~\cite{MR1097650}.

\medskip

\begin{definition}
	A \emph{partial $k$-coloring} of a graph $G$ is a sequence $(X_1,\ldots,X_k)$ of $k$ pairwise disjoint independent sets in $G$.
\end{definition}

Combining the profit structure with the notion of coloring
we define for the $k$ profit functions $p_1,\ldots,p_k: V \to \mathbb{Z}_+$ and for each partial $k$-coloring $c = (X_1,\ldots, X_k)$
a $k$-tuple $(p_1(X_1),\ldots, p_k(X_k))$, called the \emph{profit profile} of $c$.
The minimum profit of a profile, i.e., $\min_{j=1}^k \{p_j(X_j)\}$, is the \emph{satisfaction level} of $c$.
Now we can define the problem considered in this paper:

\medskip
\noindent\parbox{0.82\linewidth}{\noindent
	{\sc Fair $k$-Division Under Conflicts}\\[1.2ex]
	\begin{tabular*}{.9\textwidth}{ll}
		\noindent{\bf Input:} & A graph $G = (V,E)$, $k$ profit functions $p_1,\ldots,p_k: V \to \mathbb{Z}_+$.\\[0.5ex]
		\noindent{\bf Task:} & Compute a partial $k$-coloring of $G$ with maximum satisfaction\\& level.
	\end{tabular*}
}

\medskip
In the hardness reductions of this paper we will frequently use the decision version of this problem: for a given $q \in \mathbb{Z}_+$, does there exist a partial $k$-coloring of $G$ with satisfaction level at least $q$?

Note that an optimal partial $k$-coloring $(X_1,\ldots,X_k)$ does not necessarily select all vertices from $V$.
Furthermore, note also that for $k=1$, the problem coincides with the \textsc{Weighted Independent Set} problem: {given a graph $G = (V,E)$ and a weight function on the vertices, find an independent set of maximum total weight.}
In particular, since the case of unit weights and $k = 1$ coincides with the \textsc{Independent Set} problem, we obtain the following result.

\begin{observation}\label{obs:NP-strong-k=1}
	\textsc{Fair $1$-Division Under Conflicts} is strongly NP-hard.
\end{observation}

Thus, the addition of the conflict structure gives rise to a much more complicated problem, since \textsc{Fair $k$-Division of Indivisible Items} (which arises naturally as a special case for an edgeless conflict graph $G$) is trivial for $k=1$ and only weakly NP-hard for $k\ge 2$ (see Observation~\ref{obs:SantaNP}).

\begin{figure}[h]
	\centering
	\begin{tikzpicture}[xscale=3,yscale=1.8]
	\tikzset{roundedbox/.style={draw,rectangle,rounded corners}}
	\tikzset{cornerbox/.style={draw,rectangle}}
	\scriptsize
	\node[cornerbox] (bipartitePermutation) at (-0.7,-0.2) {\mathstrut Bipartite permutation graphs};
	\node at (-0.3,-0.4) {\mathstrut \textsf{PP}};
	\node[roundedbox] (biconvexBipartite) at (-1.2,0.8)
	{\mathstrut Biconvex bipartite graphs};
	\node at (-1.4,0.6) {\mathstrut \textsf{PP} (Thm.~\ref{thm:biconvex})};
	\node[roundedbox] (bipartite) at (-1,2.05) {\mathstrut Bipartite graphs};
	\node at (-1.36,1.85) {\mathstrut \textsf{sNPc} (Thm.~\ref{thm:bipartite})};
	\node[cornerbox] (permutation) at (-0.1,1.5) {\mathstrut Permutation graphs};
	\node at (0.1,1.3) {\mathstrut \textsf{PP}};
	\node[cornerbox] (interval) at (0.3,0.8) {\mathstrut Interval graphs};
	\node at (0.5,0.6) {\mathstrut \textsf{PP}};
	\node[roundedbox] (cocomparability) at (0.5,2.3) {\mathstrut Cocomparability graphs};
    \node at (0.05,2.1) {\mathstrut \textsf{PP} (Thm.~\ref{thm:cocomp})};
    \node[roundedbox] (chordal) at (1.5,2.35) {\mathstrut Chordal graphs};
    \node at (1.8,2.15) {\mathstrut \textsf{PP} (Thm.~\ref{thm:chordal})};
	\node[cornerbox] (comparability) at (-0.6,3) {\mathstrut Comparability graphs};
	\node at (-0.3,2.8) {\mathstrut \textsf{sNPc}};
	\node[cornerbox] (perfect) at (0.3,3.9) {\mathstrut Perfect graphs};
	\node at (-0.1,3.7) {\mathstrut \textsf{sNPc}};
	\node[roundedbox] (lineBip) at (1.7,3.2) {\mathstrut Line graphs of bipartite graphs};
	\node at (1.4,3) {\mathstrut \textsf{sNPc} (Thm.~\ref{thm:line-bipartite})};
	\node[roundedbox] (bddtw) at (1.6,0.8) {\mathstrut Graphs of bounded treewidth};
	\node at (1.8,0.6) {\mathstrut \textsf{PP} (Thm.~\ref{thm:boundedtreewidth})};
	\node[cornerbox] (forests) at (0.35,-0.2) {\mathstrut Forests};
    \node at (0.45,-0.4) {\mathstrut \textsf{PP}};
        \node[cornerbox] (edgeless) at (-0.1,-0.9) {\mathstrut Edgeless graphs};
        \node at (0.1,-1.1) {\mathstrut \textsf{PP} ({\sc Knapsack} for $k=2$)};
	\draw[semithick,-latex,font=\sffamily] (bipartitePermutation) to (biconvexBipartite);
	\draw[semithick,-latex,font=\sffamily] (biconvexBipartite) to (bipartite);
	\draw[semithick,-latex,font=\sffamily] (bipartite) to (comparability);
	\draw[semithick,-latex,font=\sffamily] (interval) to (cocomparability);
	\draw[semithick,-latex,font=\sffamily] (interval) to (chordal);
	\draw[semithick,-latex,font=\sffamily] (chordal) to (perfect);
	\draw[semithick,-latex,font=\sffamily] (permutation) to (comparability);
	\draw[semithick,-latex,font=\sffamily] (permutation) to (cocomparability);
	\draw[semithick,-latex,font=\sffamily] (comparability) to (perfect);
	\draw[semithick,-latex,font=\sffamily] (cocomparability) to (perfect);
	\draw[semithick,-latex,font=\sffamily] (bipartitePermutation) to (permutation);
	\draw[semithick,-latex,font=\sffamily] (lineBip) to (perfect);
	\draw[semithick,-latex,font=\sffamily] (forests) to (chordal);
	\draw[semithick,-latex,font=\sffamily] (forests) to (bipartite);
	\draw[semithick,-latex,font=\sffamily] (forests) to (bddtw);
        \draw[semithick,-latex,font=\sffamily] (edgeless) to (forests);
        \draw[semithick,-latex,font=\sffamily] (edgeless) to (interval);
        \draw[semithick,-latex,font=\sffamily] (edgeless) to (bipartitePermutation);
        \draw[semithick, out=20, in=270] (edgeless) to (2.42,0.7);
        \draw[semithick,-latex,font=\sffamily, out=90, in=300] (2.42,0.7) to (lineBip);
	\draw[dashed] (-1.65,1.7) .. controls (0.5,1) and (-1.3,3.8) .. (2.5,2.45);
	\end{tikzpicture}
	\caption{Relationships between various graph classes and the complexity of \textsc{Fair $k$-Division Under Conflicts} {(decision version)}. An arrow from a class $\mathcal{G}_1$ to a class $\mathcal{G}_2$ means that every graph in $\mathcal{G}_1$ is also in $\mathcal{G}_2$. Label `\textsf{PP}' means that {for each fixed $k$} the problem is solvable in pseudo-polynomial time in the given class, and label `\textsf{sNPc}' means that {for each fixed $k\ge 2$} the  decision version of the problem is strongly NP-complete.
For graph classes with round corners the result is shown in the cited theorem of this paper.
Results depicted in rectangles follow from the inclusion of graph classes.
For all graph classes in the figure,
		the problem is solvable in strongly polynomial time for $k = 1$, as it coincides with the \textsc{Weighted Independent Set} problem.}\label{fig:Hasse}
\end{figure}
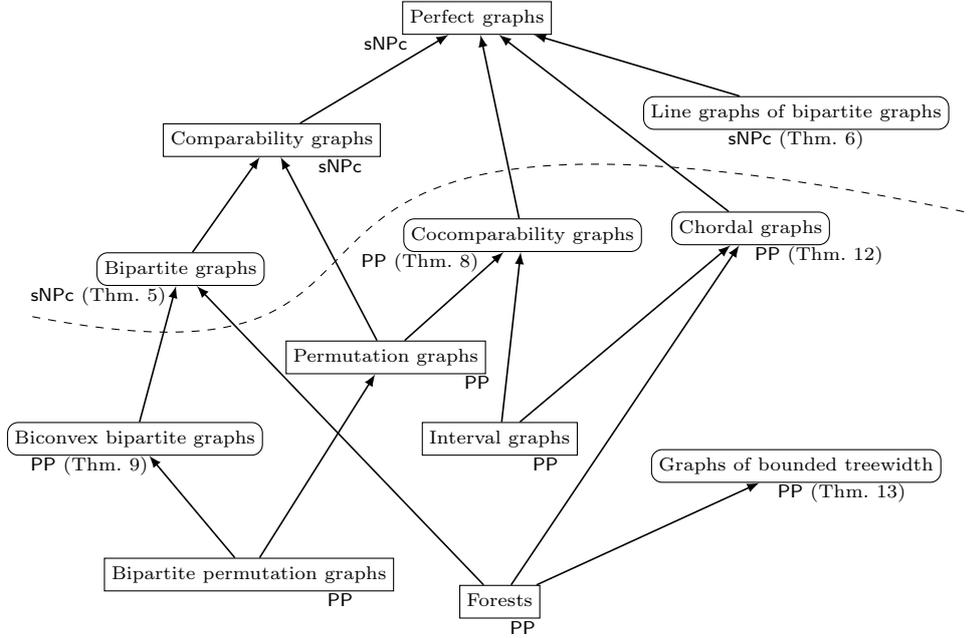

{\subsection{Our goal and contributions}}

The goal of our research is a characterization of the computational complexity of \textsc{Fair $k$-Division Under Conflicts} for different classes of conflict graphs.
We study the boundary between strongly NP-hard cases and those where a pseudo-polynomial algorithm can be derived for a constant $k$.
Observation~\ref{obs:SantaNP} implies that this is the only type of positive result we can achieve.
Moreover, considering Observation~\ref{obs:NP-strong-k=1}, it only makes sense to consider graph classes where the \textsc{Weighted Independent Set} problem is \hbox{(pseudo-)polynomially} solvable.
One such prominent example is the class of perfect graphs (see~\cite{MR936633}).
Thus, in this paper we concentrate (mainly) on various subclasses of perfect graphs as depicted in Figure~\ref{fig:Hasse}.
Additionally, we show how to adapt the algorithm for chordal graphs to obtain a pseudo-polynomial algorithm for graphs of bounded treewidth.
{
For $k=2$ our pseudopolynomial dynamic programming approaches generalize the standard dynamic program for the {\sc Knapsack} problem.}

Our contributions are as follows.
We first show that for all $k\ge 1$, the decision version of our \textsc{Fair $k$-Division Under Conflicts} is strongly NP-complete for conflict graphs from any graph class $\mathcal{G}$ for which \textsc{Independent Set} is NP-complete, provided a certain mild technical `extendability' condition is satisfied (Section~\ref{sec:generalhardness}).
By a similar reasoning we can also reach a strong inapproximability result for our problem. For bipartite conflict graphs as well as their line graphs \textsc{Fair $k$-Division Under Conflicts} can be shown to be strongly NP-hard for all $k\ge 2$ (Section~\ref{sec:bipartite}), even though the corresponding \textsc{Weighted Independent Set} problem is polynomial-time solvable.
On the other hand, for the relevant special case of biconvex bipartite graphs (cf.\ \cite{khoda13}, \cite{mast12}),
\textsc{Fair $k$-Division Under Conflicts} can be solved by a pseudo-polynomial time algorithm.
This result is based on an insightful pseudo-polynomial algorithm for the problem on a cocomparability conflict graph (Section~\ref{sec:poly}).
Besides these results, in Section~\ref{sec:poly} we present dynamic programming based solutions for the classes of chordal graphs and graphs of bounded treewidth.
{Finally, Section~\ref{sec:fptas} explains how fully polynomial time approximation schemes (FPTAS) can be derived from the pseudo-polynomial algorithms of this paper.
Figure~\ref{fig:Hasse} gives on overview of the results.}

\medskip
{\subsection{Overview of related work}}

The first elaborate treatment of the \textsc{Fair $k$-Division of Indivisible Items} problem was given in \cite{beda05}, where two approximation algorithms with non-constant approximation ratios were given.
The authors also mention that the problem cannot be approximated by a factor better than $1/2$
(under $\textrm{P} \neq \textrm{NP}$).
In~\cite{golo05} further approximation results were derived,
{among them a bicriteria approximation algorithm, which allocates a guaranteed fraction of the optimal solution value to almost all agents. }
In 2006 Bansal and Sviridenko~\cite{MR2277128} coined the term \emph{Santa Claus} problem, which corresponds to the variant of the above problem when $k$ is not fixed but part of the input.
{
Since then various approximation results have appeared on this problem of allocating indivisible items exploring different concepts of objective functions and various approximation measures, see, e.g.,~\cite{CCK09,AS10}.
}

{
An interesting variant is the {\em maximin share} concept.
Here, one considers the hypothetical scenario where every agent is allowed to partition the set of items into $k$ subsets and receives the least valued subset.
An allocation should give to every agent at least
that amount.
While this is known to be impossible in general, several approximation algorithms were derived, see \cite{KPW18,amanatidis2017approximation,BK20,GHS21}.
}

A different specialization is assumed in the widely studied \emph{Restricted Max-Min Fair Allocation} problem.
This is a special case of \textsc{Fair $k$-Division of Indivisible Items}  where every item $v_i \in V$ has a fixed valuation $p(v_i)$ and every agent either likes or ignores item $v_i$, i.e., the profit function $p_j(v_i) \in \{0, p(v_i)\}$. A fairly recent overview of approximation results both for this restricted setting as well as for the general case of the Santa Claus problem can be found in \cite{aks17}.

{Disjunctive constraints represented by conflict graphs were considered in the literature for a wide variety of combinatorial optimization problems.
Related to the allocation problem studied in this paper, there is the
knapsack problem with conflicts for which several exact algorithms were developed, most recently by \cite{Coniglio21}.
Moreover, from a similar perspective as in the current paper \cite{pfsch09,pfsch17} identified special graph classes as conflict graphs which still permit a pseudopolynomial solution algorithm.
Also the distribution of items into bins as required in the classical bin packing problem has some resemblance to \textsc{Fair $k$-Division of Indivisible Items}, where (not all) items are distributed to a fixed number of agents.
The bin packing problem with a conflict graph was studied in a number of papers, most notably in \cite{mimt10}, \cite{Sadykov13}, and \cite{Fleszar22}.
Also scheduling problems, where the allocation of jobs to machines is subject to pairwise conflicts between certain jobs, should be named as a related optimization problem.
The resulting complexity and approximation questions were considered, e.g., in \cite{boja95}, \cite{raey09}, \cite{Furman18}, and most recently in~\cite{Mallek22}.

From a more general perspective, various optimization problems on graphs were studied with the feature of an added conflict structure, e.g.,~\cite{dpsw11}, \cite{Saffari22}, and \cite{flowconflict13}.
Recently, \cite{Miao20} presented an interesting model for consistency in databases based on a conflict graph.
This widespread attention to conflict graphs in combinatorial optimization underlines the relevance of investigating disjunctive constraints also for our fair allocation problem.

The problem studied by Berge~\cite{MR989117} and de Werra~\cite{MR1097650} is similar to {\sc Fair \hbox{$k$-Division} Under Conflicts} but differs from it in one crucial aspect: instead of maximizing the minimum profit of a profile, the goal is to maximize the sum $\sum_{j=1}^k p_j(X_j)$ of all the profits.
Furthermore, they considered the case of unit profit functions $p_j:V\to \{1\}$, for all $j$, that is, the  the \textsc{Maximum Induced $k$-Colorable Subgraph} problem.
This problem has been extensively studied in the literature (see, e.g.,~\cite{MR912032,MR882643,MR3897528,MR4401492});
the case $k= 2$ is
is equivalent to the \textsc{Odd Cycle Transversal} problem
(see, e.g., \cite{MR2057781,MR4153286}).

Berge~\cite{MR989117} gave a sufficient condition for a partial $k$-coloring to be optimal, in terms of existence of a particular family of cliques, and gave several characterizations of graphs for which this condition is satisfied by every optimal solution.
Using connections with perfect graphs and balanced hypergraphs, Berge showed that line graphs of bipartite multigraphs satisfy this property.
De Werra~\cite{MR1097650} continued this line of research, applying network flow techniques and linear programming to several classes of graphs.
These characterizations rely on a min-max relation, which does not hold in general but does hold for several classes of perfect graphs (including the classes of comparability and cocomparability graphs).
The above results imply the existence of polynomial-time algorithm for the \textsc{Maximum Induced $k$-Colorable Subgraph} problem in the corresponding class of graphs, since the problem  reduces to that of finding a maximum independent set in a derived perfect graph.
Berge~\cite{MR989117} asked if for every $k$, the problem is solvable in polynomial time in the class of perfect graphs.
This is not the case unless P = NP, since Addario-Berry et al.~\cite{MR2602826} identified a subclass of perfect graphs on which the problem is NP-complete already for $k= 2$.

Due to the non-linearity of the objective function, we have no reason to expect similar min-max results for {\sc Fair $k$-Division Under Conflicts} for $k\ge 2$.
The intuition that this seems to be a much more complicated problem than \textsc{Maximum Induced $k$-Colorable Subgraph} is also confirmed by the hardness results developed in this paper, in particular, that for all $k\ge 2$ the problem is strongly $\textrm{NP}$-complete in the classes of bipartite graphs and their line graphs.}

\medskip
{\subsection{Definitions and notation}}

All graphs considered in this paper are finite, simple, and undirected. A vertex in a graph $G$ is said to be \emph{isolated} if it has no neighbors and \emph{universal} if it is adjacent to all other vertices.  A \emph{clique} in a graph $G$ is a set of pairwise adjacent vertices and an \emph{independent set} is a set of pairwise nonadjacent vertices. A \emph{matching} in $G$ is a set of pairwise disjoint edges, and a matching $M$ is \emph{perfect} if every vertex of $G$ is an endpoint of an edge of $M$. For a graph $G = (V,E)$ and a set $X\subseteq V$, we denote by $G[X]$ the \emph{subgraph of $G$ induced by $X$}, that is, the graph with vertex set $X$ in which two vertices are adjacent if and only if they are adjacent in $G$. Given two graphs $G$ and $H$, we say that $G$ is \emph{$H$-free} if no induced subgraph of $G$ is isomorphic to $H$.

\bigskip
\section{Hardness results}
\label{sec:hardness}

Observation \ref{obs:NP-strong-k=1} shows that \textsc{Fair $k$-Division Under Conflicts} is strongly NP-hard even for $k=1$ for general graphs, while Observation~\ref{obs:SantaNP} shows the weak NP-hardness of the problem for constant $k\ge 2$ in the absence of conflicts.
In what follows, we show that \textsc{Fair $k$-Division Under Conflicts} is strongly NP-hard also for all $k\ge 2$, for various well-known graph classes.

\subsection{General hardness results}
\label{sec:generalhardness}

We start with the following general property of graph classes. Let us call a graph class ${\mathcal G}$ \emph{sustainable} if every graph in the class can be enlarged in polynomial time to a graph in the class by adding to it one vertex. More formally, ${\mathcal G}$ is sustainable if there exists a polynomial-time algorithm that computes for every graph $G\in {\mathcal G}$ a graph $G'\in {\mathcal G}$ and a vertex $v\in V(G')$ such that $G'-v = G$. Clearly, any class of graphs closed under adding isolated vertices, or under adding universal vertices is sustainable. This property is shared by many well known graph classes, including planar graphs, bipartite graphs, chordal graphs, perfect graphs, etc. Furthermore, all graph classes defined by a single nontrivial forbidden induced subgraph are sustainable.

\begin{lemm}\label{lem:H-free}
For every graph $H$ with at least two vertices, the class of $H$-free graphs is sustainable.
\end{lemm}

\begin{proof}
Let ${\mathcal G}$ be the class of $H$-free graphs and let $G\in {\mathcal G}$. Since $H$ has at least two vertices, it cannot have both a universal and an isolated vertex. If $H$ has no universal vertex, then the {graph obtained from $G$ by adding to it a universal vertex} results in a graph in ${\mathcal G}$ properly extending $G$. If $H$ has no isolated vertex, then the disjoint union of $G$ with {the one-vertex graph} results in a graph in ${\mathcal G}$ properly extending $G$.
\qed\end{proof}

For an example of a graph class ${\mathcal G}$ closed under vertex deletion that is not sustainable, consider the family of all cycles and their induced subgraphs. Then every cycle is in ${\mathcal G}$ but cannot be extended to a larger graph in ${\mathcal G}$.
The importance of sustainable graph classes for \textsc{Fair $k$-Division Under Conflicts} is evident from the following theorem.

\begin{theorem}\label{thm:sustainable}
Let ${\mathcal G}$ be a sustainable class of graphs {and let $k$ be a positive integer such that} the decision version of \textsc{Fair $k$-Division Under Conflicts} is (strongly) NP-complete.
Then, for every $\ell\ge k$, the decision version of \textsc{Fair $\ell$-Division Under Conflicts} with conflict graphs from ${\mathcal G}$
is (strongly) NP-complete.
\end{theorem}

\begin{proof}
Let ${\mathcal G}$ be a sustainable class of graphs for which the decision version of \textsc{Fair $k$-Division Under Conflicts} is (strongly) NP-complete and let $\ell > k$.
Let $(G,p_1,\ldots, p_k,q)$ be an instance of \textsc{Fair $k$-Division Under Conflicts} (decision version) such that $G\in {\mathcal G}$. Since ${\mathcal G}$ is sustainable, {one can compute in polynomial time} a graph $G'\in {\mathcal G}$
such that $G'-\{x_1,\ldots, x_{\ell-k}\} = G$ for some  $\ell-k$ additional vertices $x_1,\ldots,  x_{\ell-k}$. We now define the profit functions $p_1',\ldots,p_{\ell}': V(G') \to \mathbb{Z}_+$.
For all $j = 1,\ldots, k$, let
  \begin{align*}
    \ensuremath{p'_{j}(v)=\begin{cases}
        p_{j}(v) & \text{if }\ensuremath{v\in V(G)},                                   \\
        0 & \text{if }\ensuremath{v\in\left\{x_{j}\mid 1\le j\le \ell-k\right\} .}
      \end{cases}}
  \end{align*}
and in addition let, for all $j = k+1,\ldots, \ell$, let
  \begin{align*}
    p_{j}(v)=\begin{cases}
      q & \text{if }v=x_{j-k},                                             \\
      0    & \text{if }v\in V\left(G'\right)\setminus\left\{x_{j-k}\right\}.
    \end{cases}
  \end{align*}
Observe that $G'$ has a partial $k$-coloring $(X_1',\ldots,X_k')$ such that $p_j'(X_j')\ge q$ for all $j= 1,\ldots, \ell$ if and only if $G$ has a partial $k$-coloring $(X_1,\ldots,X_k)$ such that $p_j(X_j)\ge q$ for all $j= 1,\ldots, k$. Since all the numbers involved in the reduction are polynomially bounded, we conclude that \textsc{Fair $\ell$-Division Under Conflicts} with conflict graphs from ${\mathcal G}$ is also (strongly) NP-complete.
\qed\end{proof}

{Since the \textsc{Independent Set} problem is a special case of  \textsc{Fair \hbox{$1$-Division} Under Conflicts},} Theorem~\ref{thm:sustainable} immediately implies the following.

\begin{coro}\label{cor:sustainable}
Let ${\mathcal G}$ be a sustainable class of graphs for which {the decision version of} \textsc{Independent Set} is NP-complete. Then, for every $k\ge 1$, the decision version of \textsc{Fair $k$-Division Under Conflicts} with conflict graphs from ${\mathcal G}$ is strongly NP-complete.
\end{coro}

It is known (see, e.g.,~\cite{MR765704}) that for every graph $H$ that has a component that is not a path or a subdivision of the claw (the complete bipartite graph $K_{1,3}$), the decision version of \textsc{Independent Set} is NP-complete on $H$-free graphs.
Thus, for every such graph $H$, Lemma~\ref{lem:H-free} and Corollary~\ref{cor:sustainable} imply that for every $k\ge 1$, \textsc{Fair $k$-Division Under Conflicts} (decision version) with $H$-free conflict graphs is strongly NP-complete.
{Further exploiting the relation to \textsc{Independent Set}, we also get the following strong inapproximability result for general graphs.
Its proof is closely related to the inapproximability result for \textsc{Independent Set}, but to keep the paper self-contained, we include the detailed construction in Appendix~\ref{appendixA}.
}

\begin{theorem}\label{th:inapprox}
For every $k\ge 1$ and every $\varepsilon>0$, it is $\textrm{NP}$-hard to approximate
\textsc{Fair $k$-Division Under Conflicts} within a factor of
$|V(G)|^{1-\varepsilon}$, even for unit profit functions.
\end{theorem}

\subsection{Bipartite graphs and their line graphs}\label{sec:bipartite}

In this section we show that for all $k\ge 2$, \textsc{Fair $k$-Division Under Conflicts} is $\textrm{NP}$-hard in two classes of graphs where the \textsc{Weighted Independent Set} problem is solvable in polynomial time: the classes of bipartite graphs and their line graphs.
{Recall that for a graph $H$, its line graph has a vertex for each edge of $H$, with two distinct vertices adjacent in the line graph if and only if the corresponding edges share an endpoint in $H$.
Polynomial-time solvability of the \textsc{Weighted Independent Set} problem in the class of bipartite graphs is well-known from a reduction to a network flow problem~(see, e.g.,~\cite[Corollary 21.25a]{MR1956924}).
For line graphs of bipartite graphs polynomial-time solvability follows from the facts that we can compute in linear time a bipartite graph $H$ such that the input graph $G$ is the line graph of $H$~\cite{MR347690,MR424435} and that the \textsc{Weighted Independent Set} problem on $G$ is equivalent to the weighted matching problem on $H$.
Clearly, polynomial-time solvability for the two classes also follows from the fact that both classes are subclasses of the class of perfect graphs (cf.~Figure~\ref{fig:Hasse} and~\cite[Section 66.1]{MR1956924}}).

The proof for bipartite graphs shows strong NP-hardness even for the case when all the profit functions are equal.

\begin{theorem}\label{thm:bipartite}
For each integer $k\ge 2$, the decision version of \textsc{Fair $k$-Division Under Conflicts} is strongly $\textrm{NP}$-complete in the class of bipartite graphs.
\end{theorem}

\begin{proof}
We use a reduction from the decision version of the \textsc{Clique} problem: Given a graph $G$ and an integer $\ell$, does $G$ contain a clique of size $\ell$? Consider an instance $(G,\ell)$ of \textsc{Clique} such that $2\le \ell < n := |V(G)|$. We define an instance of \textsc{Fair $k$-Division Under Conflicts} (decision version) consisting of a bipartite conflict graph $G'$, profit functions $p_1,\ldots, p_k$, and a lower bound $q$ on the required satisfaction level. The graph $G'=(A\cup B, E')$ has a vertex for each vertex of the graph $G$ as well as for each edge of $G$ and $k$ new vertices $x_1,\ldots, x_k$.
It is defined as follows:
\begin{align*}
     A&=V(G)\cup \left\lbrace x_1 \right\rbrace ,\
     B=E(G)\cup \left\lbrace x_i \mid 2\le i\le k \right\rbrace ,\\
    E'&=\left\{ ve\mid v\in V(G)\text{ is an endpoint of }e\in E(G)\right\}
     \cup
     \left\lbrace vx_{i}\mid v\in V(G),2\le i\le k\right\rbrace.
  \end{align*}
  The lower bound $q$ on the satisfaction level is defined by setting
    $q=n^4+\binom{\ell}{2}n+(n-\ell)$. For ease of notation we set $N_1 =  n^4$ and we furthermore introduce a second integer $N_2$ such that
  $q = N_{2}+\left(m-\binom{\ell}{2}\right)n$,
  where $m = |E(G)|$.  (Note that $N_2\ge n^3$.)
  With this, the profit functions \hbox{$p_i:V(G')\to\mathbb{Z}_+$}, for all $i\in \{1,\ldots, k\}$, are defined as
  $$p_i(v)=\left\{
  \begin{array}{ll}
      1;     & \hbox{if $v\in V(G)$;}\\
      n;     & \hbox{if $v\in E(G)$;}\\
      N_1;   & \hbox{if $v = x_1$;}\\
      N_2;   & \hbox{if $v = x_2$;}\\
      q ;    & \hbox{if $v = x_j$ for some $j\in \{3,\ldots, k\}$}.
  \end{array}
\right.$$
  Note that all the profits introduced as well as the number of vertices and edges of $G'$ are polynomial in $n$.
  To complete the proof, we show that $G$ has a clique of size $\ell$ if and only if $G'$ has a partial $k$-coloring with satisfaction level at least $q$.
  First assume that $G$ has a clique $C$ of size $\ell$. We construct a partial $k$-coloring $c = (X_1,\ldots, X_k)$ of $G'$ by setting
  \begin{align*}
    X_{1} & =\left\{ x_{1}\right\} \cup\left\{ e\in E(G)\mid e\subseteq C\right\} \cup\left(V(G)\setminus C\right), \\
    X_{2} & =\left\{ x_{2}\right\} \cup\left(E(G)\setminus X_{1}\right),                                         \\
    X_{j} & =\left\{ x_{j}\right\}\text{ for }3\le j\le k.
  \end{align*}
  Observe that the partial $k$-coloring $c$ gives rise to the corresponding profit profile with all entries equal to $q$, which establishes one of the two implications.

 Suppose now that there exists a partial $k$-coloring $c= (X_1,\ldots,X_k)$ of  $G'$ for which the profit profile has all entries $\geq q$. Since for each $i\in \{1,\ldots, k\}$, the total profit of the set $V(G)\cup E(G)$ is only $mn+n<n^4$,  the partial coloring $c$ must use exactly one of the $k$ vertices $x_1,\ldots, x_k$ in each color class. We may assume without loss of generality that $x_i\in X_i$ for all $i\in \{1,\ldots, k\}$. Let $U$ be the set of uncolored vertices in $G'$ w.r.t.~the partial coloring $c$. Since for each of the profit functions $p_i$, the difference between the overall sum of the profits of vertices of $G'$ and $k\cdot q$ is equal to $\ell$, we clearly have $\sum_{v\in U}p_i(v)\le \ell<n$, which implies that $U\subseteq V(G)$.
Next, observe that every vertex of $E(G)$ belongs to either $X_1$ or to $X_2$, since otherwise we would have $p_1(X_1)+p_2(X_2)<2q$,
contrary to the assumption that the satisfaction level of $c$ is at least $q$.

Consider the sets $W = X_1\cap V(G)$ and $F = X_1\cap E(G)$. Then $X_1 = \{x_1\}\cup W\cup F$ and, since $\sum_{v\in X_{1}}p_1(v)\ge q = N_{1}+\binom{\ell}{2}n+(n-\ell)$, it follows that $X_1$ contains exactly $\binom{\ell}{2}$ vertices from $E(G)$ (if $|F|>{\ell\choose 2}$, then $p_2(X_2)<q$) and at least $n-\ell$ vertices from $V(G)$. Let $C$ denote the set of all vertices of $G'$ with a neighbor in $F$. By the construction of $G'$ and since $|F| = {\ell\choose 2}$, it follows that $C$ is of cardinality at least $\ell$. Furthermore, since $X_1$ is independent, we have $C\cap W=\emptyset$.
Consequently, $n = |V(G)|\ge |C|+|W|\ge \ell + (n-\ell) = n$, hence equalities must hold throughout.
In particular, $C$ is  a clique of size $\ell$ in $G$.
\qed\end{proof}

\begin{theorem}\label{thm:line-bipartite}
For each integer $k\ge 2$, the decision version of \textsc{Fair $k$-Division Under Conflicts} is strongly $\textrm{NP}$-complete in the class of line graphs of bipartite graphs.
\end{theorem}

\begin{proof}

Note that it suffices to prove the statement for $k = 2$. For $k>2$, Theorem~\ref{thm:sustainable} applies, since the class of line graphs of bipartite graphs is sustainable.
Indeed, if $G'$ is the line graph of a bipartite graph $G$, then the graph obtained from $G'$ by adding to it an isolated vertex is the line graph of the bipartite graph obtained from $G$ by adding to it an isolated edge.

For $k = 2$, we use a reduction from the following problem:
Given a bipartite graph $G$ and an integer $Q$, does $G$ contain two disjoint matchings $M_1$ and $M_2$ such that $M_1$ is a perfect matching and $|M_2|\ge Q$? This problem was shown to be NP-complete by P\'alv\"olgi (see~\cite{palvolgyi2014partitioning}). Consider an instance $(G,Q)$ of this problem such that $1\le Q\le n/2$ and $n = |V(G)|$ is even.
Then we define the following instance of the decision version of \textsc{Fair $2$-Division Under Conflicts} with a conflict graph $G'$,
where $G'$ is the line graph of $G$.
The lower bound $q$ on the satisfaction level is defined by setting $q = n\cdot Q/2$.
The profit functions $p_1, p_2 : V(G')\to\mathbb{Z}_+$
are defined as $p_1(v)=Q$ for all $v\in V(G')$, and $p_2(v)=n/2$ for all $v\in V(G')$.
Clearly, all the profits introduced as well as the number of vertices and edges of $G'$ are polynomial in $n$.
Recall that every matching in $G$ corresponds to an independent set in $G'$.

We now show that the instances of the two decision problems have the same answers.
Suppose first that $G$ has two disjoint matchings $M_1$ and $M_2$ such that $M_1$ is a perfect matching and $|M_2|\ge Q$.
Then the sequence $(M_1,M_2)$ is a partial $2$-coloring of $G'$ such that
$$p_1(M_1) = Q|M_1| = Q\cdot n/2= q \text{ and } p_2(M_2) = (n/2)\cdot|M_2| \geq (n/2) Q = q.$$
Conversely, suppose that $G'$ has a partial $2$-coloring $(X_1,X_2)$ with satisfaction level at least $q$.
Then the independent sets $X_1$ and $X_2$ in $G'$ are disjoint matchings in $G$.
Moreover, since
$$p_1(X_1) = Q|X_1| \ge  q = Q\cdot n/2
\text{ and }
p_2(X_2) = (n/2)\cdot|X_2| \ge  q = Q\cdot n/2,$$
we obtain
$|X_1|\ge n/2$ and $|X_2| \ge Q$.
Thus, $X_1$ is a perfect matching in $G$ and
any set of $Q$ edges in $X_2$ is a matching in $G$ disjoint from $X_1$.
This proves that the decision version of \textsc{Fair $2$-Division Under Conflicts} is strongly $\textrm{NP}$-complete in the class of line graphs of bipartite graphs.
\qed\end{proof}

\section{Pseudo-polynomial algorithms for special graph classes}
\label{sec:poly}

In this section we turn our attention to classes of graphs for which the \textsc{Fair $k$-Division Under Conflicts} is solvable in pseudo-polynomial time.
As shown in Theorem~\ref{thm:bipartite}, for each $k\ge 2$, \textsc{Fair $k$-Division Under Conflicts} is strongly $\textrm{NP}$-complete in the class of bipartite graphs, and this rules out the existence of a pseudo-polynomial time algorithm for the problem in the class of bipartite graphs, unless $\textrm{P} =  \textrm{NP}$.
We show that for every $k$ there is a pseudo-polynomial time algorithm for the \textsc{Fair $k$-Division Under Conflicts} in a subclass of bipartite graphs, the class of \emph{biconvex bipartite graphs}
{(see the definition in Section~\ref{sec:biconvex})}.
The algorithm reduces the problem to the class of bipartite permutation graphs.
To solve the problem in the class of bipartite permutation graphs, we develop a solution in a more general class of graphs, the class of cocomparability graphs (containing permutation graphs).
Further, using a dynamic programming approach, we show that for every $k$ there is a pseudo-polynomial time algorithm for \textsc{Fair $k$-Division Under Conflicts} in the classes of chordal graphs
and graphs of bounded treewidth.
It will be shown in Section~\ref{sec:fptas} that all these
pseudo-polynomial dynamic programming algorithms allow the construction of a fully polynomial time approximation scheme (FPTAS).

Let us first fix some notation.
Given a graph $G$ and $k$ profit functions $p_1,\ldots, p_k:V\to \mathbb{Z}_+$, we denote by $n$ the number of vertices in $G$, $n=|V(G)|$.
All pseudo-polynomial results in this section depend on an upper bound on the maximum reachable profit value $Q=\max_{1\le j\le k}p_j(V)$.
Given an integer $k>0$, the addition and subtraction of $k$-tuples is defined component-wise, and for all $\ell \in \{1,\ldots, k\}$, we denote by $\mathbf{e}_\ell(x)$ the $k$-tuple with all coordinates equal to $0$, except that the $\ell$-th coordinate is equal to $x$.

\subsection{Cocomparability graphs}

A graph $G = (V,E)$ is a \emph{comparability graph} if it has a transitive orientation, that is, if each of the edges $\{u,v\}$ of $G$ can be replaced by exactly one of the ordered pairs $(u,v)$ and $(v,u)$ so that the resulting set $A$ of directed edges is transitive (that is, for every three vertices $x,y,z\in V$, if $(x,y)\in A$ and $(y,z)\in A$, then $(x,z)\in A$). A graph $G$ is a \emph{cocomparability graph} if its complement is a comparability graph. Comparability graphs and cocomparability graphs are well-known subclasses of perfect graphs. The class of cocomparability graphs is a common generalization of the classes of interval graphs, permutation graphs, and trapezoid graphs (see, e.g., \cite{MR1686154,MR2063679}).

Since every bipartite graph is a comparability graph, Theorem~\ref{thm:bipartite} implies that for each $k\ge 2$, \textsc{Fair $k$-Division Under Conflicts} is strongly $\textrm{NP}$-complete in the class of comparability graphs.
For cocomparability graphs, we prove that the problem is solvable in pseudo-polynomial time.
The key result in this direction is the following lemma.

\begin{lemm}\label{lem:cocomp}
For every $k\ge 1$, given a cocomparability graph $G=(V,E)$ and $k$ profit functions $p_1,\ldots, p_k:V\to \mathbb{Z}_+$, the set of all profit profiles of partial $k$-colorings of $G$ can be computed in time $\mathcal{O}(n^{k+2}(Q+1)^k)$, where $Q = \max_{1\le j\le k}p_j(V)$.
\end{lemm}

\begin{proof}
Let $G$ be a cocomparability graph. In time $\mathcal{O}(n^2)$, we compute the complement of $G$ and a transitive orientation $D$ of it~\cite{MR795937}.
Since $D$ is a directed acyclic graph, one can compute in linear time a topological sort of $D$, that is, an ordering $v_1,\ldots, v_n$ of the vertices such that if $(v_i,v_j)$ is an arc of $D$, then $i<j$
(see, e.g.,~\cite{MR2572804}).
Note that
\begin{enumerate}
\item[($\ast$)] a set $X = \{v_{i_1},\ldots, v_{i_p}\}\subseteq V$ with $i_1<\ldots<i_p$ is independent in $G$ if and only if $(v_{i_1},\ldots, v_{i_p})$ is a directed path in $D$.
\end{enumerate}
\noindent Thus, a partial $k$-coloring in $G$ corresponds to a collection of $k$ vertex-disjoint directed paths in $D$, and vice versa. We process the vertices of $G$ in the ordering given by the topological sort of $D$ and try all possibilities for the color (if any) of the current vertex $v_j$ in order to extend a partial $k$-coloring of the already processed subgraph of $G$ with $v_j$. (In terms of $D$, we choose which of the $k$ directed paths will be extended into $v_j$.) To avoid introducing additional terminology and notation, we present the details of the algorithm in terms of partial $k$-colorings of $G$ instead of systems of disjoint paths in $D$.

For each $j\in \{0,1,\ldots, n\}$ and each $k$-tuple $(i_1,\ldots, i_k)\in \{0,1,\ldots, j\}^k$, we compute the set $P_j(i_1,\ldots, i_k)$ of
all $k$-tuples $(q_1,\ldots, q_k)\in \mathbb{Z}_+^k$ such that there exists a partial $k$-coloring $(X_1,\ldots, X_k)$ of the subgraph of $G$ induced by $\{v_1,\ldots, v_j\}$ (which is empty if $j =0$)
such that $q_\ell = p_\ell(X_\ell)$ and
\begin{equation}\label{eq0}
i_\ell = \left\{
  \begin{array}{ll}
    \max\{r: v_r\in X_\ell \}, & \hbox{if  $X_\ell \neq \emptyset$;} \\
    0, & \hbox{if $X_\ell = \emptyset$}
  \end{array}
\right.
\end{equation}
for all $\ell\in \{1,\ldots, k\}$. Note that for each $\ell\in \{1,\ldots, k\}$, the possible values of the $\ell$-th coordinate of any member of $P_j(i_1,\ldots, i_k)$ belong to the set $\{0,1,\ldots, Q\}$ where $Q = \max_{1\le j\le k}p_j(V)$. Thus, each set $P_j(i_1,\ldots, i_k)$ has at most $(Q+1)^k$ elements. Note also that the total number of sets $P_j(i_1,\ldots, i_k)$ is of the order $\mathcal{O}(n^{k+1})$.

In what follows we explain how to compute the sets $P_j(i_1,\ldots, i_k)$. For $j = 0$, the only feasible choice for the $k$-tuple $(i_1,\ldots, i_k)$ is $(0,\ldots, 0)$ and we set $P_{0}(0,\ldots, 0) = \{0\}^k = \{(0,\ldots, 0)\}$. This is correct since the only partial $k$-coloring of the graph with no vertices is the $k$-tuple $(\emptyset, \ldots, \emptyset)$. Suppose that $j>1$ and that the sets $P_{j-1}(i_1,\ldots, i_k)$ are already computed for all $(i_1,\ldots, i_k)\in \{0,1,\ldots, j-1\}^k$. Fix a $k$-tuple $(i_1,\ldots, i_k)\in \{0,1,\ldots, j\}^k$. To describe how to compute the set $P_{j}(i_1,\ldots, i_k)$, we will use the following notation. We consider three cases. For each of them, we first give a formula for computing the set $P_{j}(i_1,\ldots, i_k)$ and then we argue why the formula is correct.
\begin{enumerate}
  \item If $j$ appears at least twice as a coordinate of $(i_1,\ldots, i_k)$, then we set
    \begin{equation}\label{eq1}P_{j}(i_1,\ldots, i_k) = \emptyset\,.
    \end{equation}
Note that since $j$ appears at least twice as a coordinate of $(i_1,\ldots, i_k)$, there is no partial $k$-coloring $(X_1,\ldots, X_k)$ of the subgraph of $G$ induced by $\{v_1,\ldots, v_j\}$ such that equality~\eqref{eq0} holds for all $\ell\in \{1,\ldots, k\}$. Thus, equation~\eqref{eq1} is correct.
  \item If $j$ does not appear as any coordinate of $(i_1,\ldots, i_k)$, then we set
  \begin{equation}\label{eq2}
  P_{j}(i_1,\ldots, i_k) = P_{j-1}(i_1,\ldots, i_k)\,.
  \end{equation}
  Since $j$ does not appear as any coordinate of $(i_1,\ldots, i_k)$, every partial $k$-coloring of the subgraph of $G$ induced by $\{v_1,\ldots, v_{j-1}\}$ such that equality~\eqref{eq0} holds for all $\ell\in \{1,\ldots, k\}$ is a partial $k$-coloring of the subgraph of $G$ induced by $\{v_1,\ldots, v_{j}\}$ and vice versa. This implies relation~\eqref{eq2}.
\item If $j$ appears exactly once as a coordinate of $(i_1,\ldots, i_k)$, say $i_{s} = j$, then we set
\begin{equation}\label{eq3}
P_{j}(i_1,\ldots, i_k) =\!\!\!\!
\bigcup_{\{j': j' = 0\textrm{ or } \atop v_{j'}\in N_D^-(v_j)\}}\!\!\!\!
\{\mathbf{q}+ \mathbf{e}_{s}(p_{s}(v_j)) \mid \mathbf{q} \in P_{j-1}(i_1,\dots,i_{{s}-1}, j', i_{{s}+1}, \dots, i_k)\}\,,
\end{equation}
where $N_D^-(v_j)$ denotes the set of all vertices $v_{j'}$ such that $(v_{j'},v_j)$ is an arc of $D$. (Note that $j'<j$ for all  $v_{j'}\in N_D^-(v_j)$, since $v_1,\dots, v_n$ is a topological sort of $D$.)

\begin{sloppypar}
Let $\mathbf{q}=(q_1,\ldots, q_k)\in P_{j}(i_1, \dots, i_k)$ and consider a partial $k$-coloring $(X_1,\ldots, X_k)$ of the subgraph of $G$ induced by $\{v_1,\ldots, v_j\}$ such that $p_\ell(X_\ell) = q_\ell$ and equality~\eqref{eq0} holds for all $\ell\in \{1,\ldots, k\}$. Then $\max\{q: v_q\in X_{s}\} = i_{s} = j$. In particular, $v_j\in X_{s}$. Let $X_{s}'= X_{s}\setminus\{v_j\}$ and let
\begin{equation*}
j'= \left\{
  \begin{array}{ll}
    \max\{r: v_r\in X_{s}'\}, & \hbox{if $X_{s}' \neq \emptyset$;} \\
    0, & \hbox{if $X_{s}' = \emptyset$}.
  \end{array}
\right.
\end{equation*}
Note that if $X_{s}'\neq \emptyset$ then $v_{j'}\in N_D^-(v_j)$. Indeed, digraph $D$ is an orientation of the complement of $G$, in which vertices
$v_{j'}$ and $v_j$ are adjacent (recall that they belong to the independent set $X_{s}$ in $G$). This implies that either $(v_j,v_{j'})$ or $(v_{j'},v_{j})$ is an arc of $D$, but since $j'<j$ and $v_1,\ldots, v_n$ is a topological sort of $D$, the pair $(v_{j'},v_{j})$ must be an arc of $D$.
Let $(i_1',\ldots, i_k')$ be the $k$-tuple obtained from $(i_1,\ldots, i_k)$ by replacing $i_{s}$ with $j'$, and let $(X_1',\ldots, X_k')$ be the $k$-tuple obtained from $(X_1,\ldots, X_k)$ by replacing $X_{s}$ with $X_{s}'$. Then $(X_1',\ldots, X_k')$ is a partial $k$-coloring of the subgraph of $G$ induced by $\{v_1,\ldots, v_{j-1}\}$ such that equality obtained from~\eqref{eq0} by replacing $X_\ell$ with $X_\ell'$ and
$i_{\ell}$ with $i_{\ell}'$ holds for each $\ell\in \{1,\ldots, k\}$. Furthermore, $(p_1(X_1),\ldots, p_k(X_k)) = (p_1(X_1'),\ldots, p_k(X_k'))+\mathbf{e}_{s}(p_{s}(v_j))$. This shows that if $\mathbf{q}=(q_1,\ldots, q_k)\in P_{j}(i_1,\ldots, i_k)$, then the $k$-tuple $\mathbf{q}$ belongs to the union $$\bigcup_{\{j': j' = 0\textrm{ or }v_{j'}\in N_D^-(v_j)\}}\{\mathbf{q}+ \mathbf{e}_{s}(p_{s}(v_j)) \mid \mathbf{q} \in P_{j-1}(i_1,\dots,i_{{s}-1}, j', i_{{s}+1}, \dots, i_k)\}\,.$$
\end{sloppypar}

For the converse direction, let $j' \in \{0\}\cup \{1\le j'\le j-1\mid v_{j'}\in N_D^-(v_j)\}$, let $(i_1',\ldots, i_k')$ be the $k$-tuple obtained from $(i_1,\ldots, i_k)$ by replacing $i_{s}$ with $j'$, and let $\mathbf{q}=(q_1,\ldots, q_k)\in P_{j-1}(i_1', \dots, i_k')$. Then, there exists a partial $k$-coloring $(X_1',\ldots, X_k')$ of the subgraph of $G$ induced by $\{v_1,\ldots, v_{j-1}\}$ such that for each $\ell\in \{1,\ldots, k\}$, we have $p_\ell(X_\ell') = q_\ell$ and equality obtained from~\eqref{eq0} by replacing $X_\ell$ with $X_\ell'$ and $i_{\ell}$ with $i_{\ell}'$ holds.
Let $(X_1,\ldots, X_k)$ be the $k$-tuple obtained from $(X_1',\ldots, X_k')$ by replacing $X_{s}'$ with $X_{s}'\cup\{v_j\}$.
To show that $(X_1,\ldots, X_k)$ is a partial $k$-coloring of the subgraph of $G$ induced by $\{v_1,\ldots, v_{j}\}$, it suffices to verify that $X_{s} = X_{s}'\cup\{v_j\}$ is an independent set in $G$.
If $X_{s}' = \emptyset$, then $X_{s} = \{v_j\}$ is independent.
Suppose that $X_{s}' \neq \emptyset$. Then, by ($\ast$), $X_{s}'$ corresponds to a directed path in $D$
ending in $v_{j'}$. Extending this path with vertex $v_j\in N_D^+(v_{j'})$ results in a directed path in $D$ with vertex set $X_s$, which shows, again by ($\ast$), that $X_s$ is independent in $G$.
Clearly, we have that $\max\{r: v_r\in X_{s}\} = j$, and hence $(X_1,\ldots, X_k)$ is a partial $k$-coloring of the subgraph of $G$ induced by $\{v_1,\ldots, v_{j}\}$ equality~\eqref{eq0} holds for each $\ell\in \{1,\ldots, k\}$. Furthermore, $(p_1(X_1),\ldots, p_k(X_k)) = \mathbf{q} + \mathbf{e}_{s}(p_{s}(v_j))$.
This shows that if $\mathbf{q}\in P_{j-1}(i_1', \dots, i_k')$, then the $k$-tuple $\mathbf{q}+\mathbf{e}_{s}(p_{s}(v_j))$ belongs to
$P_{j}(i_1,\ldots, i_k)$. Therefore, equation~\eqref{eq3} is correct.
\end{enumerate}

Finally, the set of all profit profiles of partial $k$-colorings of $G$ equals to the union, over all $(i_1,\ldots, i_k)\in \{0,1,\ldots, n\}^k$, of the sets $P_n(i_1,\ldots, i_k)$.

The algorithm can be easily modified so that for each profit profile also a corresponding partial $k$-coloring is computed. We would just need to store, for each $j\in \{0,1,\ldots, n\}$, each $(i_1,\ldots, i_k)\in \{0,1,\ldots, j\}^k$, and each $k$-tuple $(q_1,\ldots, q_k)\in P_j(i_1,\ldots, i_k)$, one partial $k$-coloring $(X_1,\ldots, X_k)$ of the subgraph of $G$ induced by $\{v_1,\ldots, v_i\}$ such that $p_\ell(X_\ell) = q_\ell$ and equality~\eqref{eq0} holds for all $\ell\in \{1,\ldots, k\}$.

It remains to estimate the time complexity of the algorithm. For each $j\in \{1,\ldots, n\}$ and each of the $\mathcal{O}(n^k)$ $k$-tuples $(i_1,\ldots, i_k)\in \{0,1,\ldots, j\}^k$, we can decide which of the three cases (i)--(iii) occurs in time $\mathcal{O}(k)$. Step~\eqref{eq1} takes constant time, step~\eqref{eq2} takes time $\mathcal{O}((Q+1)^k)$, and  step~\eqref{eq3} can be implemented in time  $\mathcal{O}(n(Q+1)^k)$. Altogether, this results in running time $\mathcal{O}(n(Q+1)^k)$ for each fixed $j\in \{1,\ldots, n\}$ and each $k$-tuple $(i_1,\ldots, i_k)\in \{0,1,\ldots, j\}^k$. Consequently, the total running time of the algorithm is $\mathcal{O}(n^{k+2}(Q+1)^k)$.\qed\end{proof}

Lemma~\ref{lem:cocomp} implies the following.

\begin{theorem}\label{thm:cocomp}
For every $k\ge 1$, \textsc{Fair $k$-Division Under Conflicts} is solvable in time \hbox{$\mathcal{O}(n^{k+2}(Q+1)^k)$} for cocomparability conflict graphs $G$, where $Q = \max_{1\le j\le k}p_j(V(G))$.
\end{theorem}

\begin{proof}
By Lemma~\ref{lem:cocomp}, we can compute the set $\Pi$ of all profit profiles of partial $k$-colorings of $G$ in the stated running time. For each profit profile in $\Pi$, we can determine the satisfaction level of the corresponding partial $k$-coloring of $G$. Taking the maximum satisfaction level over all profiles gives the optimal value of \textsc{Fair $k$-Division Under Conflicts} for ($G,p_1,\ldots, p_k)$.
\qed\end{proof}

\subsection{Biconvex bipartite graphs}
\label{sec:biconvex}

Recall from Theorem~\ref{thm:bipartite} that \textsc{Fair $k$-Division Under Conflicts} is strongly $\textrm{NP}$-hard for bipartite conflict graphs.
Thus, we consider in the following the more restricted case of \emph{biconvex} bipartite conflict graphs.
Recall that a bipartite graph $G=(A\cup B, E)$ is biconvex if it has a \emph{biconvex ordering}, that is, an ordering of $A$ and $B$ such that for every vertex $a \in A$ (resp.\ $b \in B$) the neighborhood $N(a)$ (resp.\ $N(b)$) is an interval of consecutive vertices
in the ordering of $B$ (resp.\ ordering of $A$).

It is known that a connected biconvex bipartite graph $G$ can always be ordered in such a way that the first and last vertices on one side have a special structure. Fix a biconvex ordering of $G$, say $A=(a_1, \ldots, a_s)$ and $B=(b_1,\ldots, b_t)$. Define $a_L$ (resp.\ $a_R$) as the vertex in $N(b_1)$ (resp.\ $N(b_t)$) whose neighborhood is not properly contained in any other neighborhood set (see \cite[Def.~8]{abbas2000biconvex}). In case of ties, $a_L$ is the smallest such index (and $a_R$ the largest). We always assume that $a_L \leq a_R$, otherwise the ordering in $A$ could be mirrored. Under these assumptions, the neighborhoods of vertices appearing in the ordering before $a_L$ and after $a_R$ are nested.

\begin{lemm}[Abbas and Stewart~\cite{abbas2000biconvex}]\label{th:biconvexstructure}
Let $G=(A\cup B, E)$ be a connected biconvex graph.
Then there exists a biconvex ordering of the vertices of $G$ such that:
\begin{enumerate}
\item  For all $a_i$, $a_j$ with $a_1 \leq a_i< a_j \leq a_L$ we have $N(a_i) \subseteq N(a_j)$.
\item For all $a_i$, $a_j$ with $a_R \leq a_i < a_j \leq a_s$  we have $N(a_j) \subseteq N(a_i)$.
\item The subgraph $G'$ of $G$ induced by vertex set $\{a_L, \ldots, a_R\} \cup B$ is a bipartite permutation graph.
\end{enumerate}
\end{lemm}

\begin{sloppypar}
Property (iii) can be put in context with Theorem~\ref{thm:cocomp}.
{Indeed, it is known that every permutation graph is a cocomparability graph} (see, e.g.,~\cite{MR1686154}).
This gives rise to the following result that \textsc{Fair $k$-Division Under Conflicts} on biconvex bipartite graphs is indeed easier (from the  complexity point of view) than on general bipartite graphs.
The high-level idea of the algorithm is illustrated in Algorithm~\ref{alg:biconvex}.
\end{sloppypar}

\begin{algorithm}\label{alg:biconvex}
	\caption{Algorithmic Idea for a Connected Biconvex Graph $G$}
	\begin{algorithmic}
	\STATE \textit{apply} Lemma~\ref{th:biconvexstructure} for getting the cocomparability graph $G'$  and vertices $a_L$, $a_R$
		\STATE let $A_L:=\{a_1,\ldots, a_{L-1}\}$ and $A_R:=\{a_{R+1}, \ldots, a_s\}$
		
		\FORALL{$j\in \{1, \ldots, k\}$}
		\STATE \textit{guess} $\overline{a}_j \in A_L$ with largest index (resp.\ smallest index $\underline{a}_j \in A_R$) included in $X_j$
		\ENDFOR
		\STATE \textit{each such guess} can be represented by a $2k$-tuple $\sigma = (\overline{a}_1, \ldots, \overline{a}_k, \underline{a}_1, \ldots, \underline{a}_k)$
		\FOR{each guess $\sigma$}
		\FORALL{$j\in \{1, \ldots, k\}$}
		\STATE \textit{exclude} all vertices $v$ of the neighborhood $N(\overline{a}_j) \subseteq B$ (and $N(\underline{a}_j) \subseteq B$)\\ from insertion into $X_j$ by setting their profit $p_j(v) := 0$
		\ENDFOR
		
		\STATE \textit{apply} Lemma~\ref{lem:cocomp} to the cocomparability graph $G'$ and the modified profit functions to obtain the set $\Pi_\sigma$ of all profit profiles $(q_1,\ldots, q_k)$ of partial $k$-colorings of $G'$ with respect to the modified profits
		
		\smallskip
		
		\STATE \textit{increase} each profit profile
		by setting $q_j := q_j + p_j(\overline{a}_j)+p_j(\underline{a}_j)$
		
		\smallskip
		
		\STATE \textit{augment} these profiles with vertices from $A_L$ and $A_R$
		
		\ENDFOR
		
		\STATE \textit{choose} the best solution over all guesses $\sigma$
		
	\end{algorithmic}
\end{algorithm}

\begin{sloppypar}
\begin{theorem}\label{thm:biconvex}
For every $k\ge 1$, \textsc{Fair $k$-Division Under Conflicts} is solvable in time \hbox{$\mathcal{O}(n^{3k+2}(Q+1)^k)$} for connected biconvex bipartite conflict graphs $G$, where $Q = \max_{1\le j\le k}p_j(V(G))$.
\end{theorem}
\end{sloppypar}

\begin{proof}
{At first Lemma~\ref{th:biconvexstructure} is applied}
for obtaining from $G$ the cocomparability graph $G'$. However, we have to consider also the vertex sets $A_L:=\{a_1,\ldots, a_{L-1}\}$
and $A_R:=\{a_{R+1}, \ldots, a_s\}$.
This is done by considering assignments of vertices in $A_L\cup A_R$ to the $k$ subsets of a partial $k$-coloring of $G$ in an efficient way as follows.

For every $j\in \{1, \ldots, k\}$, we guess, by going through all possibilities, the largest index vertex $\overline{a}_j \in A_L$ (resp.\ smallest index $\underline{a}_j \in A_R$) inserted in $X_j$.  One can add an artificial vertex $a_0$ (resp.\ $a_{s+1}$) to represent the case that no vertex from $A_L$ (resp.\ $A_R$) is inserted in $X_j$. Thus, every guess is represented by a $2k$-tuple $\sigma = (\overline{a}_1, \ldots, \overline{a}_k, \underline{a}_1, \ldots, \underline{a}_k)$.
The total number of such guesses (i.e.,\ iterations) is bounded by $(n+1)^k$ for each of $A_L$ and $A_R$, i.e., $\mathcal{O}(n^{2k})$ selections to be considered in total.

For each such guess $\sigma$ we perform the following computations. For every  $j\in \{1, \ldots, k\}$ the vertices in the neighborhood $N(\overline{a}_j) \subseteq B$ (and $N(\underline{a}_j) \subseteq B$) of the chosen index must be excluded from insertion into the corresponding set $X_j$. This can be easily realized by setting to $0$ the profits $p_j$ of all vertices in $N(\overline{a}_j)$ (resp.\ $N(\underline{a}_j)$).
With these slight modifications of the profits we can apply Lemma~\ref{lem:cocomp} for the cocomparability graph $G'$ and the modified profit functions $p_j^\sigma$ to obtain the set $\Pi_\sigma$ of all (pseudo-polynomially many) profit profiles $(q_1,\ldots, q_k)$ of partial $k$-colorings
of $G'$ with respect to $p^\sigma$. Every entry $q_j$ of a profit profile in $\Pi_\sigma$ is increased by $p_j(\overline{a}_j)+p_j(\underline{a}_j)$, to account for inclusion of the vertices selected by the guess $\sigma$.

In every guess there are the two vertices $\overline{a}_j$ and $\underline{a}_j$ permanently assigned to $X_j$ for every $j$ and their neighborhoods  $N(\overline{a}_j)$ and  $N(\underline{a}_j)$ are excluded from $X_j$.
Now it follows from properties (i) and (ii) of Lemma~\ref{th:biconvexstructure} that for each vertex $a' \in A_L$ with $a' < \overline{a}_j$ (resp.\ $a' \in A_R$ with $a'>\underline{a}_j$) the neighborhood $N(a')$ is a subset of $N(\overline{a}_j)$ (resp.\ $N(\underline{a}_j)$).
Thus, these vertices $a'$ could also be inserted in $X_j$ without any violation of the conflict structure.
Therefore, we can start from the set $\Pi_\sigma$ of profit profiles computed for $(G',p^\sigma)$ and consider iteratively (in arbitrary order) the addition of a vertex $a' \in A_L$ to one of the color classes $X_j$, as it is usually done in dynamic programming.
Each $a'$ is considered as an addition to every profit profile $(q_1,\ldots, q_{k}) \in \Pi_\sigma$ and for every index $j$ with $a' < \overline{a}_j$
yielding new profit profiles
$(q_1,\ldots, q_{j-1}, q_j+p_j(a'), q_{j+1}, \ldots, q_k)$
to be added to $\Pi_\sigma$.
An analogous procedure is performed for all vertices $a'\in A_R$
where the addition is restricted to indices $j$ with $a' > \underline{a}_j$.

For every guess $\sigma$, the running time is dominated by the effort of computing the $\mathcal{O}((Q+1)^k)$ profit profiles of $(G',p^\sigma)$ according to Lemma~\ref{lem:cocomp}, since adding any of the $\mathcal{O}(n)$ vertices $a'$ requires only $k$ operations for each profit profile.

In this way, we construct the set $\Pi_\sigma$ of all profit profiles of partial $k$-colorings of $G$ for each guess $\sigma$.
It remains to identify the optimal solution in the set $\Pi := \bigcup_\sigma\Pi_\sigma$
similarly as in the proof of Theorem~\ref{thm:cocomp}.
Going over all $\mathcal{O}(n^{2k})$ guesses $\sigma$, the total running time can be given from Lemma~\ref{lem:cocomp} as $\mathcal{O}(n^{3k+2}(Q+1)^k)$.
\qed\end{proof}

For disconnected conflict graphs, we can easily paste together the profit profiles of all connected components.
Note that this construction applies for general graphs.

\begin{lemm}
\label{thm:disconnect}
Given a conflict graph $G$ consisting of $c>1$ connected components $G_\ell$, $\ell=1, \ldots, c$, each of them with a set of profit profiles $\Pi_\ell$, where the size of each $\Pi_\ell$ is of order
$\mathcal{O}((Q+1)^k)$
with $Q = \max_{1\le j\le k}p_j(V(G))$,
\textsc{Fair $k$-Division Under Conflicts} can be solved for $G$ in time $\mathcal{O}((c-1)(Q+1)^{2k})$.
\end{lemm}

\begin{proof}
We maintain a set of profit profiles $\Pi$, initialized by $\Pi:=\Pi_1$, and iteratively merge each of the profit profiles $\Pi_2, \ldots, \Pi_m$ with $\Pi$.
To merge a set of profit profiles $\Pi_\ell$,
we consider every pair of profiles from $\Pi$ and $\Pi_\ell$
and perform a vector addition to obtain a (possibly) new profit profile which is added to $\Pi$.
At most $(Q+1)^{2k}$ such pairs may exist.
In each of the $c-1$ iterations the number of different profit profiles in $\Pi$ remains bounded by the trivial upper bound $(Q+1)^k$.
Finally, the best objective function value is determined by evaluating all profit profiles.
The total running time of this procedure is of order
$\mathcal{O}((c-1)(Q+1)^{2k})$.
\qed\end{proof}

Running Algorithm~\ref{alg:biconvex} for all $c$ components of a graph with $n$ vertices can be done in time $\mathcal{O}(n^{3k+2}(Q+1)^k)$.
Applying Lemma~\ref{thm:disconnect} on the resulting profit profiles, we obtain the following corollary.
Note that the computational complexity does not depend on the size of the components.

\begin{sloppypar}
\begin{coro}
\label{thm:biconvexdisc}
For every $k\ge 1$, \textsc{Fair $k$-Division Under Conflicts} is solvable in time \hbox{$\mathcal{O}(n^{3k+2}(Q+1)^k+(c-1)(Q+1)^{2k})$} for biconvex bipartite conflict graphs $G$ consisting of $c$ connected components,
where $Q = \max_{1\le j\le k}p_j(V(G))$.
\end{coro}
\end{sloppypar}

Note that the increased running time factor of $(Q+1)^{2k}$ cannot be easily avoided.
In particular, the natural idea of connecting the biconvex components by inserting dummy vertices to obtain a single connected biconvex graph does not work.
{This is shown in  Appendix~\ref{sec:appendix}.}

\subsection{Chordal graphs}

In this section we present a pseudo-polynomial time algorithm that solves the \textsc{Fair $k$-Division Under Conflicts} on chordal graphs. Recall that a graph is \emph{chordal} if all its induced cycles are of length three. First we state some known results on chordal graphs and their tree decompositions.

A \emph{tree decomposition} of a graph $G$ is a pair
${\cal T} = (T, \{\Bag_t\}_{t\in V(T)})$ where $T$ is a tree
whose every node $t$ is assigned a vertex
subset $\Bag_t\subseteq V(G)$ called a bag such that the following conditions are satisfied:
{\begin{itemize}
    \item Every vertex of $G$ is in at least one bag.
    \item For every edge $\{u,v\}\in E(G)$ there exists a node $t\in V(T)$ such that $\Bag_t$ contains both $u$ and $v$.
    \item For every vertex $u\in V(G)$ the subgraph of $T$ induced by the set $\{t\in V(T)\,:\,u\in \Bag_t\}$ is connected (that is, a tree).
\end{itemize}}
A tree decomposition $(T, \{\Bag_t\}_{t\in V(T)})$ is \emph{rooted} if we distinguish one vertex $r$ of $T$ which will be the root of $T$. This introduces natural parent-child and ancestor-descendant relations in the tree $T$. Following~\cite{MR3380745}, we will say
that a tree decomposition $(T, \{\Bag_t\}_{t\in V(T)})$ is \emph{nice} if it is rooted and the following
conditions are satisfied:
\begin{itemize}
\item If $t\in V(T)$ is the root or a leaf of $T$, then $\Bag_t = \emptyset$;
\item Every non-leaf node $t$ of $T$ is one of the following three types:
  \begin{itemize}
\item {\bf Introduce node:} a node $t$ with exactly one child $t'$ such that $\Bag_t = \Bag_{t'}\cup\{v\}$ for some vertex $v\in V(G)\setminus \Bag_{t'}$;
\item {\bf Forget node:} a node $t$ with exactly one child $t'$ such that $\Bag_t = \Bag_{t'}\setminus\{v\}$ for some vertex $v\in \Bag_{t'}$;
\item {\bf Join node:} a node $t$ with exactly two children $t_1$ and $t_2$ such that $\Bag_t = \Bag_{t_1} = \Bag_{t_2}$.
\end{itemize}
\end{itemize}

The \emph{width} of a tree decomposition  $(T, \{\Bag_t\}_{t\in V(T)})$ of a graph $G$ is defined as $\max_{t\in V(T)}|\Bag_t|-1$. Lemma 7.4 from~\cite{MR3380745} shows that every tree decomposition of width at most $\ell$ can be transformed in polynomial time into a nice tree decomposition of width at most $\ell$. The proof actually shows the following statement, which will be useful for our purpose.

\begin{lemm}\label{lem:nice}
Given a tree decomposition $\mathcal{T} = (T, \{\Bag_t\}_{t\in V(T)})$ of an $n$-vertex graph $G$, one can in time $\mathcal{O}(n^2\cdot\max\{n,|V(T)|\})$ compute a nice tree decomposition $\mathcal{T}'$ of $G$ that has at most $\mathcal{O}(n^2)$ nodes and such that every bag of $\mathcal{T}'$ is a subset of a bag of $\mathcal{T}$.
\end{lemm}

Let us now apply these concepts to chordal graphs. A \emph{clique tree} of a graph $G$ is a tree decomposition $(T,\{\Bag_t\}_{t\in V(T)})$ such that the bags are exactly the maximal cliques of $G$.
It is well known
(see,~e.g., \cite{MR1320296})
that a graph is chordal if and only if it has a clique tree, and in such a case a clique tree can be constructed in linear time (see, e.g.,~\cite{MR1971502}).
Furthermore, every chordal graph
$G$ has at most $|V(G)|$ maximal cliques (see, e.g.,~\cite{MR1320296}).

\begin{lemm}
\label{lem:chordal}
Given an $n$-vertex chordal graph $G$, we can compute in linear time
a tree decomposition $(T,\{\Bag_t\}_{t\in V(T)})$ of $G$ with
$\mathcal{O}(n)$ bags, all of which are cliques.
\end{lemm}

Combining Lemmas~\ref{lem:nice} and~\ref{lem:chordal} yields the following.

\begin{lemm}\label{lem:nice-chordal}
Given an $n$-vertex chordal graph $G$, we can compute in time
$\mathcal{O}(n^3)$ a nice tree decomposition $(T, \{\Bag_t\}_{t\in V(T)})$ of $G$ with $\mathcal{O}(n^2)$ bags, all of which are cliques.
\end{lemm}

We will also need the following technical lemma about tree decompositions (see, e.g.,~\cite{MR3380745}).

\begin{lemm}\label{lem:tree-dec}
Let $(T,\{\Bag_t\}_{t\in V(T)})$ be a tree decomposition of a graph $G$ and
let $\{a,b\}$ be an edge of $T$. The forest $T -\{a,b\}$ obtained from $T$ by deleting edge $\{a,b\}$ consists of two connected components $T_{a}$ (containing $a$) and $T_b$ (containing $b$). Let $A =\left(\bigcup_{t\in V(T_a)} \Bag_t\right) \setminus (\Bag_a\cap \Bag_b)$ and $B=\left(\bigcup_{t\in V(T_b)} X_t\right) \setminus (\Bag_a\cap \Bag_b)$. Then no vertex in $A$ is adjacent to a vertex in $B$.
\end{lemm}

Before we proceed to the main result for chordal graphs, we need to introduce an auxiliary definition.
Let $G = (V,E)$ be a graph, let $U\subseteq V$, let $c = (X_1,\ldots, X_k)$ be a partial $k$-coloring
of $G[X]$, and let $c' = (Y_1,\ldots, Y_k)$ be a partial $k$-coloring of $G$. We say that $c'$ \emph{agrees with $c$ on $U$} if $X_j\cap U=Y_j$ for all $j\in\{1,\dots, k\}$.

\begin{theorem}\label{thm:chordal}
\sloppypar{
For every $k\ge 1$, \textsc{Fair $k$-Division Under Conflicts} is solvable in time \hbox{$\mathcal{O}(n^{k+2}(Q+1)^{2k})$} for a chordal conflict graph $G$, where $Q = \max_{1\le j\le k}p_j(V(G))$.
}
\end{theorem}

\begin{proof}
Fix $k\ge 1$ and let $G$ be a chordal graph equipped with profit functions $p_1,\ldots, p_k:V(G)\to \mathbb{Z}_+$. We will show that we can compute the set $\Pi$ of all profit profiles of partial $k$-colorings of $G$ in the stated running time. The maximum satisfaction level over all profit profiles will then give the optimal value of \textsc{Fair $k$-Division Under Conflicts} for ($G,p_1,\ldots, p_k)$.

We first apply Lemma~\ref{lem:nice-chordal} and compute in time $\mathcal{O}(n^3)$ a nice tree decomposition $(T, \{\Bag_t\}_{t\in V(T)})$ of $G$ with $\mathcal{O}(n^2)$ bags, all of which are cliques. Recall that by definition $T$ is a rooted tree decomposition of $G$. Let $r$ be the root of $T$.
For every node $t\in V(T)$, we denote by $V_t$ the union of all bags $\Bag_{t'}$ such that $t'\in V(T)$ is a (not necessarily proper) descendant of $t$ in $T$.

We traverse tree $T$ bottom-up and use a dynamic programming approach to compute, for every node $t\in V(T)$ and every partial $k$-coloring $c$ of $G[\Bag_t]$, the family $P(t,c)$ of all profit profiles of partial $k$-colorings of $G[V_t]$ that agree with $c$ on $\Bag_t$.

Since $(T, \{\Bag_t\}_{t\in V(T)})$ is a nice tree decomposition, we have $\Bag_r = \emptyset$; in particular, the \emph{trivial} partial $k$-coloring $\emptyset^k$ consisting of $k$ empty sets is the only partial $k$-coloring of $G[\Bag_r]$. Thus, since $V_r = V(G)$ and every partial $k$-coloring of $G$ agrees with the trivial partial $k$-coloring of $G[\Bag_r]$ on $\Bag_r$, the set $P(r,\emptyset^k)$ is the set of all profit profiles of partial $k$-colorings of $G$, which is what we want to compute.

We consider various cases depending on the type of a node $t\in V(T)$ in the nice tree decomposition. For each of them we give a formula for computing the set $P(t,c)$ from the already computed sets of the form $P(t',c')$ where $t'$ is a child of $t$ in $T$, and argue why the formula is correct.

\begin{enumerate}
\item\label{step1} \textbf{$t$ is a leaf node.}

By the definition of a nice tree decomposition it follows that $\Bag_t=\emptyset$.
Thus, the only partial $k$-coloring of $G[\Bag_t]$ is the trivial one, $\emptyset^k$.
Clearly, $P(t,\emptyset^k)= \{(0,\ldots, 0)\}$.
\medskip

\item\label{step2} \textbf{$t$ is an introduce node.}

By definition, $t$ has exactly one child $t'$ and $\Bag_t = \Bag_{t'}\cup \{v\}$ holds for some vertex $v\in V\setminus X_{t'}$.
Clearly, $V_{t} = V_{t'} \cup \{v\}$, and this is a disjoint union.
(If $v\in V_{t'}$, then the subtree of $T$ consisting of all bags $\Bag_\tau$ such that $v\in \Bag_\tau$ is not connected; a contradiction.)
Consider an arbitrary partial $k$-coloring $c=(X_1,\ldots, X_k)$ of $G[\Bag_t]$. We want to compute $P(t,c)$ using the set $P(t',c')$, where $c' = (X_1\setminus\{v\},\ldots, X_k\setminus\{v\})$.
(Note that $c'$ is a partial $k$-coloring of $G[\Bag_{t'}]$.)
We claim that the following equality holds:
$$P(t,c)=\left\{
\begin{array}{ll}
\{\mathbf{q} + \mathbf{e}_j(p_j(v))\mid \mathbf{q}\in P(t',c')\}, & \hbox{if $v\in X_j$ for some $j\in \{1,\ldots, k\}$;}\\
P(t',c'), & \hbox{otherwise}.
\end{array}
\right.
$$

To show the recurrence, note first that if for all $j\in \{1,\dots, k\}$  we have $v\notin X_j$, then $c' = c$ and thus $P(t,c)=P(t',c')$ in this case. If, however, $v\in X_j$ for some $j\in \{1,\ldots, k\}$, then there can only be one such $j$, and thus $c'=(X_1,\dots, X_{j-1},X_j\setminus \{v\}, X_{j+1},\dots, X_k)$. In this case, we will need the fact that $v$ is not adjacent to any vertex of $V_{t'}\setminus \Bag_{t'}$.
Indeed, applying Lemma~\ref{lem:tree-dec} to $a = t$ and $b = t'$ shows that no vertex of
$V(G)\setminus V_{t'}$ is adjacent to any vertex of $V_{t'}\setminus \Bag_{t'}$, hence the statement follows since $v\in V(G)\setminus V_{t'}$.

The fact that all neighbors of $v$ in the set $V_{t'}$ are contained in $\Bag_{t'}$ implies that for every partial $k$-coloring of $G[V_{t'}]$ that agrees with $c'$ on $\Bag_{t'}$, adding $v$ to the $j$-th color class will result in a partial $k$-coloring of $G[V_{t}]$ that agrees with $c$ on $\Bag_t$. Thus, there is a bijective correspondence between the set of partial $k$-colorings of $G[V_{t}]$ that agree with $c$ on $\Bag_t$ and those of $G[V_{t'}]$ that agree with $c'$ on $\Bag_{t'}$, given by removing $v$ from the $j$-th color class. This implies the claimed equality $P(t,c)=\{\mathbf{q}+\mathbf{e}_j(p_j(v))\mid \mathbf{q}\in P(t',c')\}$.

\item\label{step3} \textbf{$t$ is a forget node.}

By definition, $t$ has exactly one child $t'$ in $T$ and $\Bag_{t}=\Bag_{t'}\setminus  \{v\}$ holds for some vertex $v\in V \setminus \Bag_t$. Thus, $V_t=V_{t'}$. Consider an arbitrary partial $k$-coloring $c=(X_1,\dots, X_k)$ of $G[\Bag_t]$.
We claim that the following equality holds:
$$P(t,c)=P(t',c)\cup \bigcup_{j: X_j = \emptyset}
P(t',(X_1,\dots, X_{j-1},\{v\},X_{j+1}\dots, X_k))\,.$$
Consider an arbitrary partial $k$-coloring $(Y_1,\ldots, Y_k)$ of $G[V_t]$ that agrees with $c$ on $\Bag_t$. If $v\not\in Y_j$ for all $j\in \{1,\ldots, k\}$, then $(Y_1,\ldots, Y_k)$ agrees with $c$ on $\Bag_{t'}$.
Suppose now that $v\in Y_j$ for some $j\in\{1,\ldots, k\}$. Then, $j$ is unique.
Furthermore, since $\Bag_{t'}$ is a clique in $G$ and hence in $G[V_{t'}]$, the fact that $v\in Y_j$ implies that $Y_j\cap \Bag_{t'} = \{v\}$, and consequently $X_j = Y_j\cap \Bag_{t} = \emptyset$. In this case, the partial $k$-coloring $(Y_1,\ldots, Y_k)$ agrees with the partial $k$-coloring $(X_1,\dots, X_{j-1}, \{v\}, X_{j+1}, \ldots, X_k)$ of $G[V_{t'}]$ on $\Bag_{t'}$.
Thus, every partial $k$-coloring of $G[V_t]$ that agrees with $c$ on $\Bag_t$ either agrees with $c$ on $\Bag_{t'}$ or agrees with $(X_1,\dots, X_{j-1},\{v\},X_{j+1}\dots, X_k)$ on $\Bag_{t'}$ for some $j\in \{1,\ldots, k\}$ such that $X_j = \emptyset$. Similar arguments can be used to show the converse inclusion, that is, any partial $k$-coloring of $G[V_{t'}]$ that satisfies one of the above conditions is a partial $k$-coloring of $G[V_t]$ that agrees with $c$ on $\Bag_t$. This implies the claimed equality.

\item\label{step4} \textbf{$t$ is a join node.}

By definition, $t$ has exactly two children $t_1$ and $t_2$ in $T$ and it holds that $\Bag_t = \Bag_{t_1} = \Bag_{t_2}$.
We claim that $V_{t_1}\cap V_{t_2}=\Bag_t$. It is clear that $\Bag_t\subseteq V_{t_1}\cap V_{t_2}. $ Assume for contradiction that there is a vertex $v\in V(G)$ such that $v\in (V_{t_1}\cap V_{t_2})\setminus \Bag_t$. Then there are nodes $t_1'$ and $t_2'$ of $T$ such that $v\in \Bag_{t_1'}$, $v\in \Bag_{t_2'}$, and $t_1'$ and $t_2'$ are (possibly not proper) descendants of $t_1$ and $t_2$, respectively. It follows that the subgraph of $T$ consisting of all bags containing $v$ is not connected; a contradiction. Thus $\Bag_t=V_{t_1}\cap V_{t_2}$, as claimed.
Furthermore, applying Lemma~\ref{lem:tree-dec} to $a = t_1$ and $b = t$ we can show that no vertex of $V_{t_1}\setminus \Bag_{t}$ is adjacent in $G$ to any vertex of $V(G)\setminus V_{t_1}$. Since
$V_{t_2}\setminus \Bag_{t}\subseteq V(G)\setminus V_{t_1}$, this implies that no vertex in
$V_{t_1}\setminus \Bag_{t}$ is adjacent in $G$ to any vertex of $V_{t_2}\setminus \Bag_{t}$.

Consider now an arbitrary partial $k$-coloring $c=(X_1,\dots, X_k)$ of $G[\Bag_t]$ (observe that $c$ is also a partial $k$-coloring of $G[\Bag_{t_1}]$ and $G[\Bag_{t_2}]$). In this case, we have the following recurrence relation:
$$P(t,c)=\{\mathbf{q_1}+\mathbf{q_2}-(p_1(X_1), \dots, p_k(X_k))\mid \mathbf{q_1}\in P(t_1,c), \mathbf{q_2}\in P(t_2,c)\}\,.$$

It is clear that for any partial $k$-coloring $(X_1', \ldots, X_k')$ of $G[V_t] $ that agrees with $c$ on $\Bag_t$, the $k$-tuples  $(X_1'\cap V_{t_1}, \ldots, X_k'\cap V_{t_1})$ and $(X_1'\cap V_{t_2}, \ldots, X_k'\cap V_{t_2})$ are partial $k$-colorings of $G[V_{t_1}] $ and $G[V_{t_2}]$ that agree with $c$ on $\Bag_{t_1}$ and $\Bag_{t_2}$, respectively. The fact that no vertex in $V_{t_1}\setminus \Bag_t$ is adjacent in $G$ to any vertex in $V_{t_2}\setminus \Bag_t$ implies that the other direction is also true: given partial $k$-colorings $(X_1', \ldots, X_k')$ and $(X_1'', \ldots, X_k'')$ of $G[V_{t_1}] $ and $G[V_{t_2}]$ that agree with $c$ on $\Bag_{t_1}$ and $\Bag_{t_2}$, respectively, we have $X_j'\cap \Bag_t=X_j''\cap \Bag_t=X_j$ for all $j\in\{1,\dots, k\}$, and thus $(X_1'\cup X_1'', \ldots, X_k'\cup X_k'')$ is a partial $k$-coloring of $G[V_t]$ that agrees with $c$ on $\Bag_t$.
Furthermore, for all $j\in\{1,\dots, k\}$, the fact that $V_{t_1}\cap V_{t_2}=\Bag_t$ implies that $X_j'\cap X_j''=X_j$, and hence
$p_j(X_j'\cup X_j'')=p_j(X_j')+p_j(X_j'')-p_j(X_j)$.
The claimed equality follows.
\end{enumerate}

It remains to estimate the time complexity of the algorithm. We compute a nice tree decomposition of $G$ in time $\mathcal{O}(n^3)$. Each of the $\mathcal{O}(n^2)$ bags is a clique, so in total we have $\mathcal{O}(n^k)$ partial $k$-colorings per bag.
Furthermore, note that for each partial coloring $(X_1,\ldots, X_k)$ of any induced subgraph of $G$ and each $j\in \{1,\ldots, k\}$, we have $p_j(X_j)\in \{0,1,\ldots, Q\}$. Thus, each set $P(t,c)$ has at most $(Q+1)^k$ elements.
For each of the $\mathcal{O}(n^{k+2})$ pairs $(t,c)$
where $t$ is a node of $T$ and $c$ is a partial $k$-coloring of $G[\Bag_t]$, we compute the set
$P(t,c)$ using the formula corresponding to the type of node $t$. The time complexity of this step depends on the type of the node. Case~\ref{step1} takes constant time.
In Case~\ref{step2}, we check in constant time whether $v\in X_j$ for some $j\in\{1,\dots, k\}$ and then compute the set $P(t,c)$ in time $\mathcal{O}((Q+1)^k)$.
In Case~\ref{step3}, we first compute in (constant) time $\mathcal{O}(k)$ the set of indices $j\in \{1,\ldots,k\}$ such that $X_j = \emptyset$. Then, the union given by the formula can be computed in time $\mathcal{O}((Q+1)^k)$, simply by iterating over all families in the union and keeping track of which of the  $\mathcal{O}((Q+1)^k)$ profit profiles appear in any of the families.
Finally, Case~\ref{step4} can be done in time $\mathcal{O}((Q+1)^{2k})$.
Altogether, this results in running time $\mathcal{O}((Q+1)^{2k})$ for each fixed $t\in V(T)$ and each partial $k$-coloring $c$ of $\Bag_t$. Consequently, the total running time of the algorithm is $\mathcal{O}(n^{k+2}(Q+1)^{2k})$.
\qed\end{proof}

\subsection{Graphs with bounded treewidth}
\label{sec:treewidth}

Recall that the width of a tree decomposition  $(T, \{\Bag_t\}_{t\in V(T)})$ of a graph $G$ is defined as $\max_{t\in V(T)}|\Bag_t|-1$. The \emph{treewidth} of a graph $G$
is the minimum possible width of a tree decomposition of $G$.
A graph class $\mathcal{G}$ is said to be of \emph{bounded treewidth} if there exists a nonnegative integer $\ell$ such that each graph in $\mathcal{G}$ has treewidth at most $\ell$.
{For each fixed treewidth bound $\ell$, given a graph $G$ of treewidth at most $\ell$,
a tree decomposition of $G$ of width at most $\ell$ can be computed in linear time~\cite{MR1417901}. Such a decomposition leads to linear-time algorithms for many problems that are generally NP-hard (see, e.g.,~\cite{MR1105479,MR1042649}).}

A similar approach as the one used in the proof of Theorem~\ref{thm:chordal}
for solving the \textsc{Fair $k$-Division Under Conflicts} on chordal graphs can be used on graphs of bounded treewidth.

Fix $k,\ell\ge 1$ and let $(G,p_1,\ldots, p_k)$ be the input to \textsc{Fair $k$-Division Under Conflicts} such that the treewidth of $G$ is at most $\ell$.
In time $\ell^{\mathcal{O}(\ell^3)}n$
we can compute a tree decomposition of $G$ a width at most $\ell$ using the algorithm of Bodlaender~\cite{MR1417901}. Clearly, the obtained tree decomposition has at most  $\ell^{\mathcal{O}(\ell^3)}n$ bags.
By Lemma~\ref{lem:nice} it follows that we can compute in time $\mathcal{O}(\ell^{\mathcal{O}(\ell^3)}n^3)$ a nice tree decomposition
$\mathcal{T} = (T, \{\Bag_t\}_{t\in V(T)})$ of $G$ of width at most $\ell$, with $\mathcal{O}(n^2)$ bags.
Every bag has at most $\ell+1$ vertices, so for every bag we have at most a constant number, \hbox{$(\ell+1)^{k+1}$},
partial $k$-colorings, which in total gives $\mathcal{O}(n^2)$
pairs $(t,c)$ of a node $t\in V(T)$ and a partial $k$-coloring $c$ of $t$.
For each such pair $(t,c)$, we again compute the family $P(t,c)$ of all profit profiles of partial $k$-colorings of $G[V_t]$ that agree with $c$ on $\Bag_t$.
Since $\mathcal{T}$ is a nice tree decomposition, every node is of one of the four possible types, and in Cases~\ref{step1}, \ref{step2}, and \ref{step4} we have identical equalities as in the corresponding cases in the proof of Theorem~\ref{thm:chordal}, while in Case~\ref{step3} the union over all $j$ such that $X_j=\emptyset$ of the sets $P(t',(X_1,\dots, X_{j-1},\{v\},X_{j+1}\dots, X_k))$
is replaced by the union over all $j$ such that $X_j\cup \{v\}$ is an independent set in $G$ of the sets
$P(t',(X_1,\dots, X_{j-1},X_j\cup \{v\},X_{j+1}\dots, X_k))$.
Since we can compute the adjacency matrix of $G$ in time $\mathcal{O}(n^2)$, we may assume that adjacency checks can be done in constant time.
Thus, the expressions in the formulas corresponding to each of the Cases~\ref{step2} and~\ref{step3} can be evaluated in time
$\mathcal{O}((Q+1)^k)$, while the corresponding time complexity of Case~\ref{step4} is $\mathcal{O}((Q+1)^{2k})$.
Altogether, this gives us the claimed running time and yields the following {theorem (where the constant hidden in the $\mathcal{O}$ notation depends on $k$ and $\ell$).}

\begin{theorem}\label{thm:boundedtreewidth}
For every $k\ge 1$ and $\ell\ge 1$, \textsc{Fair $k$-Division Under Conflicts} is solvable in time
$\mathcal{O}(n^2(n+(Q+1)^{2k}))$ for a graph $G$ of treewidth at most $\ell$, where $Q = \max_{1\le j\le k}p_j(V(G))$.
\end{theorem}

\section{Approximation}
	\label{sec:fptas}
	
All the pseudo-polynomial dynamic programming algorithms presented in this paper share the following characteristics.
Throughout the execution feasible states are computed, where every state describes a profit allocation given by a feasible solution of {\sc Fair $k$-Division Under Conflicts}.
Each such state is represented by a $k$-dimensional vector $(q_1, \ldots, q_k) \in \mathbb{Z}_+^k$, where every entry $q_j$ describes the profit $p_j(X_j)$ assigned to agent $j$ by a partial coloring $(X_1, \ldots, X_k)$.
While Pareto-dominated states can be eliminated,
the total number of states remains trivially bounded by $(Q+1)^k$, where $Q = \max_{1\le j\le k}p_j(V(G))$.
The optimal solution with maximum satisfaction level can be determined at the end of such an algorithm by simply going through all generated states and inspecting their satisfaction levels.

In a canonical step of our algorithms a vertex $v$ (resp.\ item) is feasibly assigned to an agent $j$ thereby generating a new state $(q_1, \ldots, q_{j-1}, q_j+p_j(v), q_{j+1}, \ldots, q_k)$ from a previous state  $(q_1, \ldots, q_k)$.
The decisions taken by the algorithms depend only on the graph but not on the profit values of previously generated states.
Every vertex is assigned to each agent at most once.

Under these preconditions, we can derive a fully polynomial time approximation scheme (FPTAS) for each such dynamic programming algorithm (considering $k$ as a constant).
For an optimal satisfaction level $z^*$, an FPTAS computes for every given $\eps>0$, an approximate solution with satisfaction {level $z^A$ fulfilling $z^A \geq z^*/(1+\eps)$} with running time polynomial in the size of the encoded input and in $1/\eps$.

The FPTAS is based on the observation that the $k$ profit values of a solution can also be seen as $k$ objective function values in a multiobjective optimization problem.
Thus, the technique for deriving an FPTAS for the multiobjective knapsack problem  described in~\cite{ekp02} can be applied as follows.

Denote the upper bound for the profit assigned to agent $j$ by $\textit{UB}_j=p_j(V(G))$
and set $u_j = \lceil n \log_{1+\eps}\textit{UB}_j\rceil$, { where, as usual, $n = |V(G)|$}.
Partition the profit range for each agent $j$ into $u_j$ intervals
\begin{align*}
[1, (1+\eps)^{1/n}), \:
[(1+\eps)^{1/n}, (1+\eps)^{2/n}), \:
&[(1+\eps)^{2/n}, (1+\eps)^{3/n}), \ldots\\
&[(1+\eps)^{(u_j-1)/n}, (1+\eps)^{u_j/n}]\,.
\end{align*}
To obtain an FPTAS from the generic dynamic programming algorithm indicated above we
restrict the possible profit values $q_j$ allocated to agent $j$ to the lower interval endpoints of these intervals.
The FPTAS mimics exactly the operations of the exact dynamic program, but whenever a vertex $v$ is assigned to $j$, the resulting profit $q_j+p_j(v)$ is {\em rounded down} to the nearest interval endpoint.
Note that this does not change the steps of the dynamic program since we assumed that its decisions do not depend on the profit values of states.

{The bound $u_j = \lceil n \log_{1+\eps}\textit{UB}_j\rceil$ is in $\mathcal{O}(n /\eps\cdot \log_2(\textit{UB}_j))$, which is polynomial in the encoding length of the input, since $$\log_{1+\eps}\textit{UB}_j=
(\ln2 \log_2 \textit{UB}_j)/\ln(1+\eps)
\leq
(2\ln2 \log_2 \textit{UB}_j)/\eps\,,$$ for all $\eps \in (0,1)$. The above inequality follows from \hbox{$x\le 2\ln(1+x)$}, which can be verified to hold for all $x \in (0,1)$ by standard calculus.}
Thus, the total number of states in the modified algorithm is bounded by $\mathcal{O}((n/\eps)^k (\log_2 Q)^k)$.

Concerning the loss of accuracy we can proceed similarly to~\cite{ekp02} and compare an arbitrary state $(q_1, q_2, \ldots, q_k)$ of the exact dynamic program to some state of the FPTAS consisting of lower interval endpoints $(\tq_1, \tq_2, \ldots, \tq_k)$.
For every state $(q_1, \ldots, q_j, \ldots, q_k)$ generated by the exact algorithm after assigning $i$ vertices to agent $j$, we claim that in the FPTAS there exists a state $(\tq_1, \tq_2, \ldots, \tq_k)$ of lower interval endpoints such that
\begin{equation}\label{eq:fptasrecur}
q_j \leq (1+\eps)^{i/n} \tq_j\,.
\end{equation}
This claim can be shown by induction.
For $i=1$, there was one vertex $v$ assigned to agent $j$ giving profit $q_j=p_j(v)$.
In the FPTAS, there will be a state where $\tq_j$ is the largest lower interval endpoint not exceeding $q_j$.
By construction of the intervals, we have $(1+\eps)^{1/n} \tq_j \geq q_j$.

Assuming the claim to be true for some $i-1$,
we consider the $i$-th assignment of a vertex $v$ to $j$.
In the exact algorithm, $p_j(v)$ is added to some value $q_j$
for which there exists a lower interval endpoint $\tq_j$ fulfilling
$q_j \leq (1+\eps)^{(i-1)/n} \tq_j$.
During the FPTAS, $p_j(v)$ will also be added to $\tq_j$ and the result will be rounded down to a lower interval endpoint $\tq'$ with
$(1+\eps)^{1/n} \tq' \geq \tq_j + p_j(v) \geq (1+\eps)^{-(i-1)/n} q_j + p_j(v)
\geq  (1+\eps)^{-(i-1)/n} (q_j + p_j(v)) $.
Moving terms around, this proves (\ref{eq:fptasrecur}) for the new profit $q_j + p_j(v)$.

Since there can be at most $n$ vertices assigned to any agent, (\ref{eq:fptasrecur}) holds also for the satisfaction level of the optimal solution.

Summarizing the above discussion and the proofs of Theorem~\ref{thm:cocomp}, Corollary~\ref{thm:biconvexdisc}, Theorem~\ref{thm:chordal}, and Theorem~\ref{thm:boundedtreewidth}, we conclude:

\begin{sloppypar}
\begin{theorem}\label{thm:fptas}
\textsc{Fair $k$-Division Under Conflicts} with constant $k$ admits an FPTAS if the conflict graph is
a cocomparability graph, a biconvex bipartite graph, a chordal graph, or a graph of bounded treewidth.
\end{theorem}
\end{sloppypar}

To put Theorem~\ref{thm:fptas} in perspective, recall that by~Theorem~\ref{th:inapprox} no constant-factor approximation for
\textsc{Fair $k$-Division Under Conflicts} exists for general graphs, unless P = NP.

	\section{Conclusions}
	\label{sec:conc}
	In this paper we introduced the \textsc{Fair $k$-Division Under Conflicts} and studied it from a computational complexity point of view, with respect to various restrictions on the conflict graph. In particular, we could show that the problem is strongly $\textrm{NP}$-hard on general bipartite conflict graphs, but can be solved in pseudo-polynomial time on biconvex bipartite graphs, on chordal graphs, on cocomparability graphs, and on graphs of bounded treewidth.
	There are other graph classes sandwiched between the two classes of our results, for which the complexity of \textsc{Fair $k$-Division Under Conflicts} is still open.
	In particular, we can derive open problems from the following sequence of inclusions:
	biconvex bipartite $\subseteq$
	convex bipartite $\subseteq$
	interval bigraph $\subseteq$
	chordal bipartite $\subseteq$ bipartite.
	{We believe that a positive result for convex bipartite graphs could be within reach.}
	Outside this chain of inclusions, we pose the complexity of the problem for planar bipartite conflict graphs as another interesting open question.

\appendix

\section{Proof of Theorem \ref{th:inapprox}}\label{appendixA}

Fix an integer $k\ge 1$. We give a reduction from the
\textsc{Independent Set} problem.
We construct a graph $G'$ by taking $k$ copies of $G$ and by adding all possible edges between vertices from different copies. Furthermore we take $k$ {``unit''} profit functions $p_1,\dots, p_k$ from $V(G')$ to $\{1\}$. We claim that the maximum size of an independent set in $G$ equals the maximum satisfaction level of a partial $k$-coloring in $G'$ (with respect to the profit functions $p_1,\dots, p_k$). Given a maximum independent set $I$ in $G$ of size $q$ one can immediately obtain a partial $k$-coloring $(X_1,\dots,X_k)$ of $G'$ with satisfaction level $q$ by inserting all vertices of $I$ in the $j$-th copy of $G$ into $X_j$, for all $j=1,\ldots,k$. On the other hand, given a partial $k$-coloring $(X_1,\dots,X_k)$ of $G'$ with satisfaction level $q$, one can simply choose $X_1$, which is an independent set completely contained in one copy of $G$. Thus, $X_1$ corresponds to an independent set in $G$ of size $q$.

Suppose that for some $\varepsilon\in (0,1)$ there exists a polynomial-time algorithm $A$ that approximates \textsc{Fair $k$-Division Under Conflicts} within a factor of $|V(G)|^{1-\varepsilon}$ on input instances with unit profit functions. We will show that this implies the existence of a polynomial-time algorithm $A'$  approximating the \textsc{Independent Set} problem within a factor of $|V(G)|^{1-\varepsilon'}$ where $\varepsilon'= \varepsilon/2$. As shown by Zuckerman~\cite{MR2403018}, this would imply $\textrm{P} = \textrm{NP}$.

Consider an input graph $G$ to the \textsc{Independent Set} problem.
The algorithm $A'$ proceeds as follows. If $|V(G)|<k^{2(1-\varepsilon)/\varepsilon}$, then the graph is of constant order and the problem can be solved optimally in $\mathcal{O}(1)$ time. If $|V(G)|\ge k^{2(1-\varepsilon)/\varepsilon}$, then the graph $G'$ is constructed following the above reduction, a partial $k$-coloring $(X_1,\ldots, X_k)$ is computed using algorithm $A$ on $G'$ equipped with $k$ unit profit functions, and a subset of $V(G)$ corresponding to $X_1$ is returned. Clearly, the algorithm runs in polynomial time and computes an independent set in $G$. Let $q$ denote the maximum satisfaction level of a partial $k$-coloring in $G'$. By the above claim, the independence number of $G$ equals $q$.
Thus, to complete the proof, it suffices to show that $|X_1|\ge q/(|V(G)|^{1-\varepsilon'})$. By assumption on $A$, we have that $|X_1|\ge  q/(|V(G')|^{1-\varepsilon})$.
We want to show that
$q/|V(G')|^{1-\varepsilon}\ge q/|V(G)|^{1-\varepsilon'}$, or, equivalently,
$1/k^{1-\varepsilon}|V(G)|^{1-\varepsilon} \ge 1/|V(G)|^{1-\varepsilon/2}$.
After some straightforward algebraic manipulations, this inequality simplifies to
the equivalent inequality $|V(G)|\ge k^{2(1-\varepsilon)/\varepsilon}$, which is true by assumption.\qed

\section{A remark on biconvex graphs}
\label{sec:appendix}

Biconvex bipartite graphs were characterized by forbidden induced subgraphs by Tucker in~\cite{MR295938}.
The list of forbidden induced subgraphs includes all cycles except the cycle of length four and five additional graphs, including the two graphs $F_1$ and $F_2$ depicted in Figure~\ref{fig:fis}.

\begin{figure}[h!]
  \centering
   \includegraphics[width=0.7\textwidth]{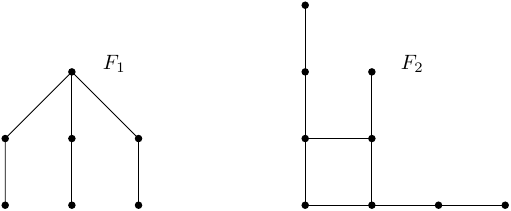}
\caption{Two forbidden induced subgraphs for biconvex bipartite graphs.}
\label{fig:fis}
\end{figure}

\begin{proposition}\label{prop:disconnected}
There exists a disconnected biconvex bipartite graph that is not an induced subgraph of any connected
biconvex bipartite graph.
\end{proposition}
\begin{sloppypar}
\begin{proof}
Consider the graph $G$ depicted in Figure~\ref{fig:example}.

\begin{figure}[h!]
  \centering
   \includegraphics[width=0.7\textwidth]{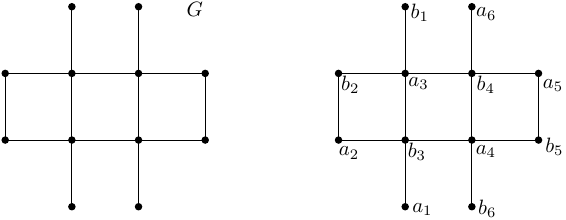}
\caption{A $12$-vertex biconvex bipartite graph and a biconvex labeling of it.}
\label{fig:example}
\end{figure}

As shown by the vertex labeling in the figure, $G$ is a biconvex bipartite graph.
Consequently, the graph $G+K_2$, the disjoint union of $G$ and a complete graph of order two, is also a biconvex bipartite graph.
We will show that $G+K_2$ is not an induced subgraph of any connected biconvex bipartite graph.

Fix a labeling of $G$ as in Figure~\ref{fig:example}, take a disjoint copy of $K_2$, call it $G'$, and suppose for a contradiction that the disjoint union $G+G'$ is an induced subgraph of a connected biconvex bipartite graph $H$.
Let $A$ and $B$ denote the two parts of a bipartition of $H$ so that
$\{a_1,\ldots,a_6\} \subseteq A$ (and then $\{b_1,\ldots,b_6\} \subseteq B$).

Since $H$ is connected, it contains a path from $V(G')$ to $V(G)$. Let $P$ be a shortest such path.
Since the sets $V(G)$ and $V(G')$ are disjoint and the are no edges between them, $P$ has at least three vertices.
Let $x$ be the only vertex on $P$ that has a neighbor in $G$, let $y$ be the neighbor of $x$ on $P$ such that $y\not\in V(G)$,
and let $z$ be defined as follows:
$$z = \left\{
  \begin{array}{ll}
    \text{the neighbor of $y$ on $P$ other than $x$}, & \hbox{if $P$ has at least $4$ vertices;} \\
    \text{the neighbor of $y$ in $G'$}, & \hbox{if $P$ has exactly three vertices.}
  \end{array}
\right.$$
Since $H$ is bipartite, it contains no cycle of length three. This implies that vertices $x$ and $z$ are not adjacent
to each other.

By symmetry of $G$, we may assume that $x\in A$ (and thus $y\in B$ and $z\in A$).
Furthermore, by the minimality of $P$, vertices $y$ and $z$ do not have any neighbors in $V(G)$.
We make a series of observations about the neighborhood of $x$ in $V(G)$.
\begin{itemize}
  \item Vertex $x$ cannot be adjacent to both $b_3$ and $b_4$, since otherwise $H$ would contain an induced $F_1$ with vertex set
$\{x,y,z,b_3,a_2,b_4,a_5\}$.

By symmetry, we may assume that $x$ is not adjacent to $b_4$.

  \item Vertex $x$ is not adjacent to $b_5$.
Suppose that it is. Then $x$ is not adjacent to $b_3$, since
otherwise the set $\{x,b_3,a_3,b_4,a_5,b_5\}$ would induce a $6$-cycle in $H$.
But now, $H$ contains an induced $F_1$ with vertex set
$\{x,b_5,a_4,b_3,a_2,b_4,a_6\}$, a contradiction.

  \item Vertex $x$ is adjacent to $b_3$.
Suppose that this is not the case.
Then $x$ is not adjacent to $b_i$ for $i\in \{1,2\}$, since otherwise $H$ would contain an induced $F_1$ with vertex set
$\{x,b_i,a_3,b_3,a_1,b_4,a_5\}$.
Therefore, the only possible neighbor of $x$ in $V(G)$ is $b_6$.
But now, $H$ contains an induced $F_1$ with vertex set
$\{x,b_6,a_4,b_3,a_1,b_4,a_5\}$, a contradiction.

\item Vertex $x$ is adjacent to $b_2$, since otherwise $H$ would contain an induced $F_1$ with vertex set
$\{y,x,b_3,a_2,b_2,a_4,b_5\}$.
\end{itemize}
To conclude the proof, we observe that $H$ contains an induced $F_2$ with vertex set
$\{z,y,x,b_2,a_3,b_3,a_1,b_4,a_5\}$, a contradiction.
\qed
\end{proof}
\end{sloppypar}

\medskip
{\footnotesize\noindent{\it Acknowledgements.}
The work of this paper was done in the framework of two bilateral projects between University of Graz and University of Primorska, financed by the OeAD (SI 22/2018 and SI 31/2020) and the Slovenian Research Agency (BI-AT/18-19-005 and BI-AT/20-21-015).
		The authors acknowledge partial support of the Slovenian Research Agency (I0-0035, research programs P1-0285, P1-0383, and P1-0404, research projects N1-0102, N1-0160, N1-0210, J1-3001, J1-3002, J1-3003, J1-4008, and J5-4596, and a Young Researchers Grant)
		and the European Commission for funding the InnoRenew CoE project (Grant Agreement \#739574) under the Horizon2020 Widespread-Teaming program and the Republic of Slovenia (Investment funding of the Republic of Slovenia and the European Union of the European
		Regional Development Fund)
		and by the Field of Excellence ``COLIBRI'' at the University of Graz and by the Federal Ministry for Digital and Economic Affairs of the Republic of Austria through the COIN project FIT4BA.}


\begin{thebibliography}{10}

\bibitem{abbas2000biconvex}
N.~Abbas and L.~K. Stewart.
\newblock Biconvex graphs: ordering and algorithms.
\newblock {\em Discrete Applied Mathematics}, 103(1-3):1--19, 2000.

\bibitem{MR2602826}
L.~Addario-Berry, W.~S. Kennedy, A.~D. King, Z.~Li, and B.~Reed.
\newblock Finding a maximum-weight induced {$k$}-partite subgraph of an
  {$i$}-triangulated graph.
\newblock {\em Discrete Applied Mathematics}, 158(7):765--770, 2010.

\bibitem{MR765704}
V.~E. Alekseev.
\newblock The effect of local constraints on the complexity of determination of
  the graph independence number.
\newblock In {\em Combinatorial-Algebraic Methods in Applied Mathematics},
  pages 3--13. Gorky University Press, 1982.
\newblock in Russian.

\bibitem{amanatidis2017approximation}
G.~Amanatidis, E.~Markakis, A.~Nikzad, and A.~Saberi.
\newblock Approximation algorithms for computing maximin share allocations.
\newblock {\em ACM Transactions on Algorithms}, 13(4):52, 2017.

\bibitem{aks17}
C.~Annamalai, C.~Kalaitzis, and O.~Svensson.
\newblock Combinatorial algorithm for restricted max-min fair allocation.
\newblock {\em ACM Transactions on Algorithms}, 13(3), 2017.

\bibitem{MR1105479}
S.~Arnborg, J.~Lagergren, and D.~Seese.
\newblock Easy problems for tree-decomposable graphs.
\newblock {\em Journal of Algorithms}, 12(2):308--340, 1991.

\bibitem{AS10}
A.~Asadpour and A.~Saberi.
\newblock An approximation algorithm for max-min fair allocation of indivisible
  goods.
\newblock {\em SIAM Journal on Computing}, 39(7):2970--2989, 2010.

\bibitem{azep98}
Y.~Azar and L.~Epstein.
\newblock On-line machine covering.
\newblock {\em Journal of Scheduling}, 1:67--77, 1998.

\bibitem{MR2277128}
N.~Bansal and M.~Sviridenko.
\newblock The {S}anta {C}laus problem.
\newblock In {\em S{TOC}'06: {P}roceedings of the 38th {A}nnual {ACM}
  {S}ymposium on {T}heory of {C}omputing}, pages 31--40. 2006.

\bibitem{BK20}
S.~Barman and S.~K. Krishnamurthy.
\newblock Approximation algorithms for maximin fair division.
\newblock {\em ACM Transactions on Economics and Computation}, 8(1), 2020.

\bibitem{BLM21}
X.~Bei, X.~Lu, P.~Manurangsi, and W.~Suksompong.
\newblock The price of fairness for indivisible goods.
\newblock {\em Theory of Computing Systems}, pages 1--25, 2021.

\bibitem{MR989117}
C.~Berge.
\newblock Minimax relations for the partial {$q$}-colorings of a graph.
\newblock {\em Discrete Mathematics}, 74(1-2):3--14, 1989.

\bibitem{beda05}
I.~Bezakova and V.~Dani.
\newblock Allocating indivisible goods.
\newblock {\em ACM SIGecom Exchanges}, 5(3):11--18, 2005.

\bibitem{MR1320296}
J.~R.~S. Blair and B.~Peyton.
\newblock An introduction to chordal graphs and clique trees.
\newblock In {\em Graph theory and sparse matrix computation}, volume~56 of
  {\em IMA Vol. Math. Appl.}, pages 1--29. Springer, New York, 1993.

\bibitem{boja95}
H.~Bodlaender and K.~Jansen.
\newblock On the complexity of scheduling incompatible jobs with unit-times.
\newblock In {\em MFCS '93: Proceedings of the 18th International Symposium on
  Mathematical Foundations of Computer Science}, pages 291--300. Springer,
  1993.

\bibitem{MR1417901}
H.~L. Bodlaender.
\newblock A linear-time algorithm for finding tree-decompositions of small
  treewidth.
\newblock {\em SIAM Journal on Computing}, 25(6):1305--1317, 1996.

\bibitem{ijcai2017-20}
S.~Bouveret, K.~Cechlárová, E.~Elkind, A.~Igarashi, and D.~Peters.
\newblock Fair division of a graph.
\newblock In {\em Proceedings of the Twenty-Sixth International Joint
  Conference on Artificial Intelligence, {IJCAI-17}}, pages 135--141, 2017.

\bibitem{hand16}
S.~Bouveret, Y.~Chevaleyre, and N.~Maudet.
\newblock Fair allocation of indivisible goods.
\newblock In F.~Brandt, V.~Conitzer, U.~Endriss, J.~Lang, and A.~D. Procaccia,
  editors, {\em Handbook of Computational Social Choice}, pages 284--310.
  Cambridge University Press, 2016.

\bibitem{MR1686154}
A.~Brandst\"{a}dt, V.~B. Le, and J.~P. Spinrad.
\newblock {\em Graph classes: a survey}.
\newblock SIAM Monographs on Discrete Mathematics and Applications. Society for
  Industrial and Applied Mathematics (SIAM), 1999.

\bibitem{Brito21}
S.~S. Brito and H.~G. Santos.
\newblock Preprocessing and cutting planes with conflict graphs.
\newblock {\em Computers and Operations Research}, 128:105176, 2021.

\bibitem{CCK09}
D.~Chakrabarty, J.~Chuzhoy, and S.~Khanna.
\newblock On allocating goods to maximize fairness.
\newblock In {\em Proceedings - Annual IEEE Symposium on Foundations of
  Computer Science, FOCS}, pages 107--116, 2009.

\bibitem{iwoca2020}
N.~Chiarelli, M.~Krnc, M.~Milani\v{c}, U.~Pferschy, N.~Piva\v{c}, and
  J.~Schauer.
\newblock Fair packing of independent sets.
\newblock In {\em Combinatorial Algorithms - 31st International Workshop,
  {IWOCA} 2020}, volume 12126 of {\em LNCS}, pages 154--165. Springer, 2020.

\bibitem{Coniglio21}
S.~Coniglio, F.~Furini, and P.~{San Segundo}.
\newblock A new combinatorial branch-and-bound algorithm for the knapsack
  problem with conflicts.
\newblock {\em European Journal of Operational Research}, 289(2):435--455,
  2021.

\bibitem{MR2572804}
T.~H. Cormen, C.~E. Leiserson, R.~L. Rivest, and C.~Stein.
\newblock {\em Introduction to algorithms}.
\newblock MIT Press, Cambridge, MA, third edition, 2009.

\bibitem{MR1042649}
B.~Courcelle.
\newblock The monadic second-order logic of graphs. {I}. {R}ecognizable sets of
  finite graphs.
\newblock {\em Information and Computation}, 85(1):12--75, 1990.

\bibitem{MR3380745}
M.~Cygan, F.~V. Fomin, {\L}.~Kowalik, D.~Lokshtanov, D.~Marx, M.~Pilipczuk,
  M.~Pilipczuk, and S.~Saurabh.
\newblock {\em Parameterized algorithms}.
\newblock Springer, 2015.

\bibitem{MR4153286}
K.~K. Dabrowski, C.~Feghali, M.~Johnson, G.~Paesani, D.~Paulusma, and
  P.~Rz\k{a}\.{z}ewski.
\newblock On cycle transversals and their connected variants in the absence of
  a small linear forest.
\newblock {\em Algorithmica}, 82(10):2841--2866, 2020.

\bibitem{dpsw11}
A.~Darmann, U.~Pferschy, J.~Schauer, and G.~Woeginger.
\newblock Paths, trees and matchings under disjunctive constraints.
\newblock {\em Discrete Applied Mathematics}, 159:1726--1735, 2011.

\bibitem{MR1097650}
D.~de~Werra.
\newblock Packing independent sets and transversals.
\newblock In {\em Combinatorics and graph theory}, volume~25 of {\em Banach
  Center Publ.}, pages 233--240. PWN, Warsaw, 1989.

\bibitem{DFL1982}
B.~L. Deuermeyer, D.~K. Friesen, and M.~A. Langston.
\newblock Scheduling to maximize the minimum processor finish time in a
  multiprocessor system.
\newblock {\em SIAM Journal on Algebraic and Discrete Methods}, 3(2):190--196,
  1982.

\bibitem{ekp02}
T.~Erlebach, H.~Kellerer, and U.~Pferschy.
\newblock Multiobjective knapsack problems.
\newblock {\em Management Science}, 48:1603--1612, 2002.

\bibitem{raey09}
G.~Even, M.~M. Halld{\'o}rsson, L.~Kaplan, and D.~Ron.
\newblock Scheduling with conflicts: online and offline algorithms.
\newblock {\em Journal of Scheduling}, 12(2):199--224, 2009.

\bibitem{Factorovich20}
P.~Factorovich, I.~Méndez-Díaz, and P.~Zabala.
\newblock Pickup and delivery problem with incompatibility constraints.
\newblock {\em Computers and Operations Research}, 113:104805, 2020.

\bibitem{Fleszar22}
K.~Fleszar.
\newblock A {MILP} model and two heuristics for the bin packing problem with
  conflicts and item fragmentation.
\newblock {\em European Journal of Operational Research}, 303(1):37--53, 2022.

\bibitem{Furman18}
H.~Furmańczyk and M.~Kubale.
\newblock Scheduling of unit-length jobs with cubic incompatibility graphs on
  three uniform machines.
\newblock {\em Discrete Applied Mathematics}, 234:210--217, 2018.

\bibitem{MR912032}
F.~Gavril.
\newblock Algorithms for maximum {$k$}-colorings and {$k$}-coverings of
  transitive graphs.
\newblock {\em Networks}, 17(4):465--470, 1987.

\bibitem{GHS21}
M.~Ghodsi, M.~T. Hajiaghayi, M.~Seddighin, S.~Seddighin, and H.~Yami.
\newblock Fair allocation of indivisible goods: Improvement.
\newblock {\em Mathematics of Operations Research}, 46(3):1038--1053, 2021.

\bibitem{golo05}
D.~Golovin.
\newblock Max-min fair allocation of indivisible goods.
\newblock Technical Report CMU-CS-05-144, Carnegie Mellon University, 2005.

\bibitem{MR2063679}
M.~C. Golumbic.
\newblock {\em Algorithmic {G}raph {T}heory and {P}erfect {G}raphs}, volume~57
  of {\em Annals of Discrete Mathematics}.
\newblock Elsevier, second edition, 2004.

\bibitem{MR936633}
M.~Gr\"{o}tschel, L.~Lov\'{a}sz, and A.~Schrijver.
\newblock {\em Geometric algorithms and combinatorial optimization}, volume~2
  of {\em Algorithms and Combinatorics: Study and Research Texts}.
\newblock Springer-Verlag, Berlin, 1988.

\bibitem{Hu15}
Z.-H. Hu, J.-B. Sheu, L.~Zhao, and C.-C. Lu.
\newblock A dynamic closed-loop vehicle routing problem with uncertainty and
  incompatible goods.
\newblock {\em Transportation Research Part C: Emerging Technologies},
  55:273--297, 2015.

\bibitem{khoda13}
K.~Khodamoradi, R.~Krishnamurti, A.~Rafiey, and G.~Stamoulis.
\newblock {PTAS} for ordered instances of resource allocation problems.
\newblock In {\em Proceedings of the 33rd International Conference on
  Foundations of Software Technology and Theoretical Computer Science, FSTTCS
  2013}, volume~24 of {\em LIPICS}, pages 461--473, 2013.

\bibitem{KPW18}
D.~Kurokawa, A.~D. Procaccia, and J.~Wang.
\newblock Fair enough: Guaranteeing approximate maximin shares.
\newblock {\em Journal of the ACM}, 65(2), 2018.

\bibitem{MR4401492}
O.~Kuryatnikova, R.~Sotirov, and J.~C. Vera.
\newblock The maximum {$k$}-colorable subgraph problem and related problems.
\newblock {\em INFORMS J. Comput.}, 34(1):656--669, 2022.

\bibitem{MR347690}
P.~G.~H. Lehot.
\newblock An optimal algorithm to detect a line graph and output its root
  graph.
\newblock {\em Journal of the Association for Computing Machinery},
  21:569--575, 1974.

\bibitem{Mallek22}
A.~Mallek and M.~Boudhar.
\newblock Scheduling on uniform machines with a conflict graph: complexity and
  resolution.
\newblock {\em International Transactions in Operational Research}, to appear,
  2022.

\bibitem{mast12}
M.~Mastrolilli and G.~Stamoulis.
\newblock Restricted max-min fair allocations with inclusion-free intervals.
\newblock In {\em Proceedings of International Computing and Combinatorics
  Conference COCOON 2012}, volume 7434 of {\em LNCS}, pages 98--108. Springer,
  2012.

\bibitem{Miao20}
D.~Miao, Z.~Cai, J.~Li, X.~Gao, and X.~Liu.
\newblock The computation of optimal subset repairs.
\newblock {\em Proceedings of the VLDB Endowment}, 13(12):2061--2074, 2020.

\bibitem{MR3897528}
N.~Misra, F.~Panolan, A.~Rai, V.~Raman, and S.~Saurabh.
\newblock Parameterized algorithms for max colorable induced subgraph problem
  on perfect graphs.
\newblock {\em Algorithmica}, 81(1):26--46, 2019.

\bibitem{mimt10}
A.~Muritiba, M.~Iori, E.~Malaguti, and P.~Toth.
\newblock {Algorithms for the bin packing problem with conflicts}.
\newblock {\em INFORMS Journal on Computing}, 22(3):401--415, 2010.

\bibitem{palvolgyi2014partitioning}
D.~P{\'a}lv{\"o}lgi.
\newblock Partitioning to three matchings of given size is {NP}-complete for
  bipartite graphs.
\newblock {\em Acta Universitatis Sapientiae, Informatica}, 6(2):206--209,
  2014.

\bibitem{pfsch09}
U.~Pferschy and J.~Schauer.
\newblock The knapsack problem with conflict graphs.
\newblock {\em Journal of Graph Algorithms and Applications}, 13(2):233--249,
  2009.

\bibitem{flowconflict13}
U.~Pferschy and J.~Schauer.
\newblock The maximum flow problem with disjunctive constraints.
\newblock {\em Journal of Combinatorial Optimization}, 26(1):109--119, 2013.

\bibitem{pfsch17}
U.~Pferschy and J.~Schauer.
\newblock Approximation of knapsack problems with conflict and forcing graphs.
\newblock {\em Journal of Combinatorial Optimization}, 33(4):1300--1323, 2017.

\bibitem{MR2057781}
B.~Reed, K.~Smith, and A.~Vetta.
\newblock Finding odd cycle transversals.
\newblock {\em Operations Research Letters}, 32(4):299--301, 2004.

\bibitem{MR424435}
N.~D. Roussopoulos.
\newblock A max {$\{m,n\}$} algorithm for determining the graph {$H$} from its
  line graph {$G$}.
\newblock {\em Information Processing Letters}, 2:108--112, 1973.

\bibitem{Sadykov13}
R.~Sadykov and F.~Vanderbeck.
\newblock Bin packing with conflicts: A generic branch-and-price algorithm.
\newblock {\em INFORMS Journal on Computing}, 25(2):244--255, 2013.

\bibitem{Saffari22}
S.~Saffari and Y.~Fathi.
\newblock Set covering problem with conflict constraints.
\newblock {\em Computers \& Operations Research}, 143:105763, 2022.

\bibitem{Santos19}
L.~F.~M. Santos, R.~S. Iwayama, L.~B. Cavalcanti, L.~M. Turi, F.~E.
  de~Souza~Morais, G.~Mormilho, and C.~B. Cunha.
\newblock A variable neighborhood search algorithm for the bin packing problem
  with compatible categories.
\newblock {\em Expert Systems with Applications}, 124:209--225, 2019.

\bibitem{MR1956924}
A.~Schrijver.
\newblock {\em Combinatorial optimization. {P}olyhedra and efficiency.},
  volume~24 of {\em Algorithms and Combinatorics}.
\newblock Springer, 2003.

\bibitem{MR795937}
J.~Spinrad.
\newblock On comparability and permutation graphs.
\newblock {\em SIAM Journal on Computing}, 14(3):658--670, 1985.

\bibitem{MR1971502}
J.~P. Spinrad.
\newblock {\em Efficient graph representations}, volume~19 of {\em Fields
  Institute Monographs}.
\newblock American Mathematical Society, Providence, RI, 2003.

\bibitem{MR295938}
A.~Tucker.
\newblock A structure theorem for the consecutive {$1$}'s property.
\newblock {\em Journal of Combinatorial Theory. Series B}, 12:153--162, 1972.

\bibitem{MR1452078}
G.~J. Woeginger.
\newblock A polynomial-time approximation scheme for maximizing the minimum
  machine completion time.
\newblock {\em Operations Research Letters}, 20(4):149--154, 1997.

\bibitem{MR882643}
M.~Yannakakis and F.~Gavril.
\newblock The maximum {$k$}-colorable subgraph problem for chordal graphs.
\newblock {\em Information Processing Letters}, 24(2):133--137, 1987.

\bibitem{MR2403018}
D.~Zuckerman.
\newblock Linear degree extractors and the inapproximability of max clique and
  chromatic number.
\newblock {\em Theory of Computing}, 3:103--128, 2007.

\end{thebibliography}
\end{document}